\documentclass[11pt, onecolumn, journal]{IEEEtran}
\usepackage{epsfig, amsmath,amssymb,epsf,cite,subfigure}

\def\cA{{\mathcal{A}}}  \def\cC{{\mathcal{C}}} 
   
  \def\cK{{\mathcal{K}}} \def\cL{{\mathcal{L}}}
 \def\cN{{\mathcal{N}}}  
   
  \def\cW{{\mathcal{W}}} 
 \def\cZ{{\mathcal{Z}}}

\def\bf{{\mathbf{f}}}  \def\bh{{\mathbf{h}}}  
   \def\bn{{\mathbf{n}}} 
\def\bp{{\mathbf{p}}}   \def\bs{{\mathbf{s}}} 
\def\bu{{\mathbf{u}}} \def\bv{{\mathbf{v}}}  \def\bx{{\mathbf{x}}} \def\by{{\mathbf{y}}}
\def\bz{{\mathbf{z}}}

\def\bA{{\mathbf{A}}} \def\bB{{\mathbf{B}}}   
 \def\bG{{\mathbf{G}}} \def\bH{{\mathbf{H}}} \def\bI{{\mathbf{I}}} 
   \def\bN{{\mathbf{N}}} 
\def\bP{{\mathbf{P}}}    \def\bT{{\mathbf{T}}}
  \def\bW{{\mathbf{W}}}  \def\bY{{\mathbf{Y}}}
\def\bZ{{\mathbf{Z}}}

\def\argmin{\mathop{\mathrm{argmin}}}
\def\argmax{\mathop{\mathrm{argmax}}}
\def\det{\mathop{\mathrm{det}}}

\def\tr{\mathop{\mathrm{tr}}}

\def\dim{\mathop{\mathrm{dim}}}

\def\rank{\mathop{\mathrm{rank}}}
\def\span{\mathop{\mathrm{span}}}

     \def\d4{\!\!\!\!}              \def\ep{\epsilon}
                  \def\gam{\gamma}     
               
\def\bSig{\mathbf{\Sigma}} \def\bsh{\backslash} \def\bGam{\mathbf{\Gamma}}
 \def\bPi{\mathbf{\Pi}}

  \def\R{{\mathbb{R}}} \def\C{{\mathbb{C}}} 

\def\bzero{\mathbf{0}}

\def\lp{\left(}     \def\rp{\right)}    \def\ls{\left\{}    \def\rs{\right\}}    \def\lS{ \left[ }
\def\rS{ \right] }  \def\la{\left|}     \def\ra{\right|}    \def\lA{\left\|}     \def\rA{\right\|}

  \def\barm{\bar{m}}

    \def\bbA{\bar{\bA}}
\def\bbB{\bar{\bB}}

      \def\tbh{\tilde{\bh}}  
   \def\tbH{\widetilde{\bH}}
\def\tby{\tilde{\by}} \def\tbz{\tilde{\bz}}  \def\tbs{\tilde{\bs}} \def\tbA{\widetilde{\bA}}
\def\tbB{\widetilde{\bB}} \def\tbN{\widetilde{\bN}}  \def\tbp{\tilde{\bp}}
\def\tR{\tilde{R}}
   \def\hbB{\widehat{\bB}} \def\hu{\hat{u}}
\def\hk{\hat{k}}
  \def\-{\! - \!}  \def\+{\! + \!}  \def\={\! = \!}  \def\>{\! > \!}

\newtheorem{theorem}{Theorem}
\newtheorem{lemma}{Lemma}
\newtheorem{corollary}{Corollary}

\newtheorem{remark}{Remark}

\newtheorem{definition}{Definition}

\newcommand{\bet}{\begin{table}}
\newcommand{\eet}{\end{table}}
\newcommand{\btt}{\begin{tabular}}
\newcommand{\ett}{\end{tabular}}
\newcommand{\bec}{\begin{center}}
\newcommand{\eec}{\end{center}}
\newcommand{\bef}{\begin{figure}}
\newcommand{\eef}{\end{figure}}
\newcommand{\beq}{\begin{eqnarray}}
\newcommand{\eeq}{\end{eqnarray}}
\newcommand{\bea}{\begin{array}}
\newcommand{\eea}{\end{array}}

\newenvironment{proof}[1][Proof]{\begin{trivlist}
\item[\hskip \labelsep {\bfseries #1}]}{\end{trivlist}}

\newcommand{\qed}{\nobreak \ifvmode \relax \else
\ifdim\lastskip<1.5em \hskip-\lastskip
\hskip1.5em plus0em minus0.5em \fi \nobreak
\vrule height0.5em width0.6em depth0.1em\fi}

\begin{document}

\title{On the Spatial Degrees of Freedom of Multicell and Multiuser MIMO Channels}

\author{Taejoon~Kim,~\IEEEmembership{Student~Member,~IEEE,}
        David~J.~Love,~\IEEEmembership{Senior~Member,~IEEE,}
        and~Bruno~Clerckx,~\IEEEmembership{Member,~IEEE,}
\thanks{T. Kim and D. J. Love are with the School
of Electrical and Computer Engineering, Purdue University, West
Lafayette,
IN, 47906 USA (e-mail: kim487@ecn.purdue.edu, djlove@ecn.purdue.edu).}
\thanks{B. Clerckx is with Samsung Advanced Institute of Technology,
Samsung Electronics, Yongin-Si, Gyeonggi-Do, 446-712 Korea
(e-mail: bruno.clerckx@samsung.com).}
\thanks{This work was supported in part by Samsung Electronics.}}

\maketitle
\begin{abstract}
We study the converse and achievability for the degrees of freedom
of the multicellular multiple-input multiple-output (MIMO) multiple access channel (MAC) with
constant channel coefficients.
We assume $L>1$ homogeneous cells with $K\geq 1$ users per cell where the users have $M$ antennas and
the base stations are equipped with $N$ antennas.
The degrees of freedom outer bound for this $L$-cell and $K$-user MIMO MAC is formulated.
The characterized outer bound uses insight from a limit on the total degrees of freedom for the $L$-cell heterogeneous MIMO network.
We also show through an example that a scheme selecting a transmitter and performing partial message sharing
outperforms a multiple distributed transmission strategy in terms of the total degrees of freedom.
Simple linear schemes attaining the outer bound (i.e., those achieving the optimal degrees of freedom) are explores for a few cases.
The conditions for the required spatial dimensions attaining the optimal degrees of freedom
are characterized in terms of $K$, $L$, and the number of transmit streams.
The optimal degrees of freedom for the two-cell MIMO MAC are examined by using transmit zero forcing and null space interference alignment
and subsequently, simple receive zero forcing is shown to provide the optimal degrees of freedom for $L>1$.
By the uplink and downlink duality, the degrees of freedom results in this paper are also applicable to the downlink.
In the downlink scenario, we study the degrees of freedom of $L$-cell MIMO interference channel exploring multiuser diversity.
Strong convergence modes of the instantaneous degrees of freedom as the number of users  increases are characterized.
\end{abstract}
\IEEEpeerreviewmaketitle
\setlength{\baselineskip}{21pt}

\section{Introduction}

Over the past few years, a significant amount of research has gone into making various techniques for
enhancing spectrum reusability reality. Spatial  techniques such as multiple-input
multiple-output (MIMO) wireless systems have been widely studied to improve the spectrum reusability.
Recently, the scope of spatial transmission has been extended to MIMO network wireless systems such as the
interference network, relay network, and multicellular network.
Network MIMO systems are now an emphasis of IMT-Advanced and beyond systems.
In these networks, out-of-cell (or cross cell) interference is a major drawback.
Before network MIMO can be deployed and used to its full potential, there are a large number of challenging issues.
Many of these deal with interference management and joint processing between nodes to suppress out-of-cell
interference (e.g., see the references in \cite{Gesbert1}).

\subsection{Overview}

Understanding the information-theoretic capacity of general network MIMO is still
challenging even under full cooperation assumptions.
Alternatively, there are various approaches to approximate the capacity in the high SNR regime (some of which can
be practically achieved in small cell scenarios \cite{Gesbert1}) by analyzing the number of resolvable
interference-free signal dimensions in terms of the degrees of freedoms
of the network. Initial works include the degrees of freedom and/or capacity region characterization for the MIMO
multiple access channel (MAC) \cite{Tse1} and MIMO broadcast channel \cite{vishwanath,viswanath,Yu, weingarten}.
While the general capacity region of the interference channel is not known,
there are some known capacity results with \emph{very strong} \cite{Carleial}  and \emph{strong} \cite{Han, Sato} interference.
The capacity  outer bounds \cite{Carleial1, Kramer} and degrees of freedom outer bounds \cite{Host, Host1} for the multiple nodes
interference channel with single antenna nodes have been characterized.
Recently, the degrees of freedom have been studied for the two node MIMO X channel \cite{Maddah, Jafar2} and
the two user MIMO interference channel \cite{Jafar1}. The key innovation used to prove the inner bound on the degrees
of freedom is interference alignment \cite{Jafar2, Jafar1}.

Interference alignment aims to allow coordinated transmission and reception  in order to
increase the total degrees of freedom of the network.
Interference alignment generates overlapping user signal spaces occupied by undesired interference
while keeping the desired signal spaces distinct.
When an achievable scheme achieves the degrees of freedom of the converse, we say that the scheme attains the \textit{optimal degrees of freedom}.

The fundamental idea of interference alignment in \cite{Jafar2, Jafar1} is extended to the multiple node X channel in \cite{Cadambe1},
$K$-user interference channel in \cite{Cadambe2, Gou}, and more general cellular networks in \cite{Suh}
under a time or frequency varying channel assumption.
For the X channel with single antenna users, interference alignment achieves the optimal degrees of freedom
for the $K$ by $L\=2$ (or $K\=2$ by $L$) X channel with finite symbol extension, but
for $K>2$ and $L>2$, it requires infinite symbol extension \cite{Cadambe1}.
The $K$-user interference channel with single antenna nodes \cite{Cadambe2} and multiple antenna nodes \cite{Gou}
also needs infinite symbol extension.
Various aspects of interference alignment for cellular networks are investigated in \cite{Suh} including
the effect of a multi-path channel and channel with propagation delay.
The work in \cite{Suh} shows that a single degree of freedom can be achieved per user as the number of users
grows large with symbol extension.

In the case of constant channel coefficients, the spatial degrees of freedom have mainly been
investigated. For the two by two MIMO X channel, the exact optimal degrees of freedom
of $\frac{4}{3}M$ is achievable when each node has $M>1$ antennas \cite{Jafar2, Maddah}.
The optimal degrees of freedom of the two user MIMO interference channel is shown to be
$\min\lp 2M,2N, \max(M,N) \rp$ in \cite{Jafar1}, where $M$ and $N$ denote the number antennas at the transmitter and
receiver, respectively.
Remarkably, simple zero forcing is sufficient to provide the optimal degrees of freedom \cite{Jafar2, Jafar1}.
Interference alignment in a three-user interference channel with $M=N$ antennas at each node yields
the optimal degrees of freedom of $\frac{3M}{2}$ when $M$ is even (when $M$ is odd a two
symbol extension is required to achieve $\frac{3M}{2}$) \cite{Cadambe2}.
Compared to the prototypical examples of the two-user MIMO interference channel or two by two MIMO X channel,
the general characterization of the optimal degrees of freedom for the multicell multiuser MIMO networks
(that works for an arbitrary numbers of users and cells) with constant channel coefficients is still an open problem.
When studying the achievable scheme with constants channel coefficients, the number of required $M$ and $N$
must be determined as a function of the number of cells ($L$) and users ($K$) or vice versa.
Thus, taking into consideration all of these dependencies often makes the characterization overconstrained.
Recently, an achievable scheme where each user obtains one degree of freedom for the two cell and $K$-user MIMO
network with constant channel coefficients is proposed for $N=M=K+1$ in \cite{Suh1}.
In an $L$-cell and $K$-user MIMO network, a necessary zero interference condition on $M$ and $N$
(as a function of $K$ and $L$) to provide one interference free dimension to each of users
is investigated in \cite{Honig1}.

The conventional interference alignments and other linear schemes in \cite{Jafar2,Jafar1,Cadambe1,Cadambe2,Gou,Suh}
require global notion of CSI at all nodes, and
the optimal degrees of freedom is particularly attained by extending signals over large space/time/frequency dimensions.
To overcome these challenges, efficient interference alignment schemes that only utilize local CSI feedback are
considered in \cite{Gomadam,Suh1}.
An efficient way to provide additional degrees of freedom gain without a global notion of CSI and, at the same time,
with a reduced amount of feedback is to exploit multiuser diversity as in \cite{viswanath1,Sharif}.
The basic notion of the multiuser diversity with multiple antennas in \cite{viswanath1,Sharif} has been
recently extended to interference networks, namely through opportunistic interference alignment,
such as for the case of a cognitive network \cite{Perlaza}, cellular uplink \cite{Bangchul},
and cellular downlink \cite{Tang, Gesbert2, Junghoon}.
The common idea is to schedule users (or dimensions in \cite{Perlaza}) so that the interference
caused by the selected users to the other receivers are aligned or minimized with the aid of power
allocation \cite{Perlaza, Gesbert2} and opportunistic transmit or receive filter design
\cite{Perlaza, Bangchul,  Tang,  Junghoon}. The performance of the multiuser diversity is
evaluated or analyzed in terms of the average throughput \cite{Perlaza, Tang, Gesbert2} and
average degrees of freedom \cite{Bangchul, Junghoon}.

\subsection{Contributions}
First, a simple characterization of the optimal degrees of freedom with constant channel coefficient
for the multicell MIMO MAC is provided. 
Then, a scenario when the downlink system exploits the multiuser diversity is considered and 
the degrees of freedom by employing user scheduling is characterized.  

In the uplink, we assume $L$ homogeneous cells with $K$ users per cell.
We do not consider time or frequency domain extensions with a time or frequency varying channel assumption.
Alternatively, spatial resources are utilized with constant channel coefficients.
Although our focus is on the scenario where the transmitter and receiver have $M$ and $N$ antennas, we show a spatial degrees of
freedom outer bound for the $L$-cell and $K$-user MIMO MAC that includes the case when each node has a different number of antennas.
For the two-cell case, two linear schemes that achieve the degrees of freedom outer bound are characterized.
The first scheme is a simple transmit zero forcing with $M=K\beta + \beta$ and $N=K\beta$, and the second
one is a null space interference alignment with $M=K\beta$ and $N=K\beta+\beta$, where $\beta>0$ is a positive integer.
For $L>1$ (including the two-cell case), it is verified that receive zero forcing with $M=\beta$ and $N=KL\beta$ precisely
achieves the optimal degrees of freedom for $K\geq 1$.

The main ingredients of the degrees of freedom outer bound, analogous to \cite{Host,Cadambe1,Cadambe2,Gou}, are to split whole messages into small subsets
so that the outer bound can tractably be formulated for each of message subsets.
We define the message subset for the $L$-cell heterogeneous networks where
$L-1$ cells form an $L-1$-user MIMO interference channel and a single cell forms a $K$-user MIMO MAC.
We also investigate through an example that selecting a subset of transmitters and allowing them to use partial
message sharing (through perfect links) achieves a higher degrees of freedom than distributed MIMO transmission.

Null space interference alignment for the two-cell case is developed for the uplink scenario with $N>M$
to show the achievability of the converse.
It relies on each base station using a carefully chosen null space plane.
The null space planes are designed to project the out-of-cell interference to a lower dimensional subspace than its original
dimension so that the null space plane can jointly mitigate the degrees of freedom loss coming from
the out-of-cell interference.
The dimensions of the interference free signal at each base station
after projection depend on the ``size" of the overlapped out-of-cell interference null space, which is referred
to as the \emph{geometric multiplicity} of the out-of-cell interference null space (the definition will be clearer
in Section \ref{section5}). We generalize the null space interference alignment framework for various kinds
of antenna dimensions. Though it does not necessarily achieve the optimal degrees of freedom,
it resolves $\beta>0$ interference free dimensions per user.
Notice that by the uplink and downlink duality the degrees of freedom results obtained for the uplink are also applicable to the downlink.

Next, we study the degrees of freedom of the $L$-cell downlink interference channel by exploiting multiuser diversity.
One of the key aspects for the interference alignment in \cite{Cadambe1, Cadambe2, Gou, Suh} is in its
almost sure (a.s.) convergence argument on the instantaneous degrees of freedom with infinite
symbol extension across time and frequency.
In line with the convergence argument made in interference alignments, we show that this strong
convergence argument on the instantaneous degrees of freedom still holds when utilizing many users
in the network.
We quantify the additional degrees of freedom achievable through the user scheduling
where the user scheduling only uses the local CSI.
This exhibits clear comparison on the instantaneous degrees of freedom between
the multiuser diversity system and interference alignment in \cite{Cadambe1, Cadambe2, Gou, Suh}.
We show in particular that if the number of candidate users that participate in scheduling
in a cell increases faster than linearly with SNR, the instantaneous degrees of freedom
converges to $L$ in both mean-square (m.s.) sense and almost sure (a.s.) sense for the $L$-cell
downlink MIMO interference channel with $M=1$ and $N=L-1$.

The rest of the paper is outlined as follows.
Section \ref{section2} describes the system model.
In Section \ref{section3}, the degrees of freedom outer bound for $L$-cell and $K$-user MIMO MAC is formulated.
The conditions for the optimal degrees of freedom are characterized in Section \ref{section4}.
In Section \ref{section5}, general frameworks for the null space interference alignment for various kinds of spatial dimension
conditions are investigated.
Section \ref{section6} discusses the instantaneous degrees of freedom with multiuser diversity for the $L$-cell
downlink MIMO interference channel.
The paper is concluded in Section \ref{section_conlusions}.

\section{System Model} \label{section2}

We first define the uplink channel model.
The downlink channel model is simply described by the uplink and downlink duality.
\subsection{Uplink Channel Model}

Consider a network that consists of $L$ homogeneous cells. In each cell, there are $K\geq 1$ users
and one base station, where each user has $M\geq 1$ antennas and the base station is
equipped with $N\geq 1$ antennas. We introduce an index $\ell k$ to correspond to user $k$ in cell $\ell$ for
$\ell\in\cL$ and $k\in\cK$ where $\cL=\{1,\ldots, L\}$ and $\cK=\{1,\ldots, K\}$, respectively.
For instance, a $3$-cell MIMO MAC is shown in Fig. \ref{Fig1} where each cell consists of $2$ users (i.e., $L=3$ and $K=2$).
Note that though our focus, in this paper, is on $L$ homogeneous cells where the transmitter and receiver have $M$ and $N$ antennas,
respectively, we generalize the degrees of freedom outer bound when user $\ell k$ has
$M_{\ell k}$ antenna and base station $\ell$ has $N_{\ell}$ antennas in Section \ref{subsection3.1}.

The channel input-output relation at the $t$th discrete time slot is described as
\beq
\by_{m}(t)= \sum_{\ell=1}^{L}\sum_{k=1}^{K}\bH_{m,\ell k}\bx_{\ell k}(t)+\bz_{m}(t), \  m \in \cL \label{channel_model1}
\eeq
where $\by_{m}(t)\in\C^{N\times 1}$ and $\bz_{m}(t)\in \C^{N\times 1}$ denote the received signal vector and additive noise vector
at the base station $m$, respectively.
Each entry of $\bz_m(t)$ is independent and identically distributed (i.i.d.) with
$\cC\cN(0,1)$.
The vector $\bx_{\ell k}(t)\in\C^{M\times 1}$ in \eqref{channel_model1} represents the user $\ell k$'s transmit vector at $t$th channel use.
The channel input is subject to an individual power constraint
\beq
E\lS \lA\bx_{\ell k}(t)\rA^2 \rS=\tr\lp E\lS \bx_{\ell k}(t)\bx_{\ell k}^*(t) \rS \rp \leq \rho, \  k\in\cK,  \ell\in\cL \label{power_const1}
\eeq
where $\rho$ represents SNR.
The matrix $\bH_{m,\ell k}\in\C^{N\times M}$ in \eqref{channel_model1} denotes the channel with constant coefficients from
user $\ell k$ to base station $m$.
Moreover, $\ls \bH_{m,mk}\rs_{k\in\cK}$ represent the desired data channels at base station $m$
while the matrices $\ls \bH_{m,\ell k} \rs_{\ell \in\cL\bsh m, k\in\cK}$ carry out-of-cell interference to  base station $m$.
All the channel matrices are sampled from continuous distributions, and each entry of $\bH_{m,\ell k}$ is i.i.d. (i.e., we
basically assume a rich scattering environment).
This channel model almost surely ensures all channel matrices have full rank, i.e.,
$\footnote{Throughout the paper, the $\rank(\bA)$ for $\bA\in\C^{N \times M}$ extracts a dimension of the range space of
$\bA$, i.e., $\rank(\bA)=\dim(ran(\bA))$, where the range space is defined as
$ran(\bA)=\{ \by\in\C^{N\times 1}: \by=\bA\bx, \bx\in\C^{M\times 1} \}$ and $\dim(\cA)$ extracts the number
of basis of the subspace $\cA$.
Null space of $\bA$ is defined as $null(\bA)=\{ \bx\in\C^{M\times 1}: \bA\bx=\bzero \}$.}
\rank(\bH_{m,\ell k})=\min(M,N)$ for $m,\ell\in\cL$ and $k\in\cK$.
The channel gains from different users are mutually independent.
This channel condition where all channel matrices with i.i.d. are full rank is referred to as \emph{nondegenerate} in this paper.

Define $W_{\ell k}(\rho)$ as a message from user $\ell k$ to the destined base station $\ell$ at SNR $\rho$.
The message $W_{\ell k}(\rho)$ is uniformly distributed in a $(n,2^{nR_{\ell k}(\rho)})$ codebook
$\cZ(\rho)$$=$$\{ \zeta_{1}(\rho), \ldots, \zeta_{2^{nR_{\ell k}(\rho)}}(\rho) \}$, and messages at different
users are independent of each other. In order to approach the capacity, the data rate of the
coding scheme increases with respect to (w.r.t) $\rho$. This includes a coding \emph{scheme} where the codebook is
chosen from a sequence of codebooks $\ls \cW(\rho) \rs$ for each level of $\rho$.
The message $W_{\ell k}(\rho)$ is mapped to $\bx_{\ell k}(t)$ in \eqref{channel_model1} over $n$ channel uses.
Then, the information transfer rate $R_{\ell k}(\rho)$ of message $W_{\ell k}(\rho)$ is
said to be achievable if the probability of decoding error can be made arbitrarily small by choosing
an appropriate channel block length $n$.
The capacity region $\cC(\rho)$ is the set of all achievable rate tuples $\{R_{\ell k}(\rho)\}_{\ell \in\cL, k\in\cK}$.

\subsection{Degrees of Freedom} \label{subsection2.2}

We define the spatial degrees of freedom of the multicell MIMO MAC as
\beq
\bSig_{d} = \lim_{\rho\rightarrow\infty}
                                        \sum_{\{ R_{\ell k}(\rho) \}_{\ell\in\cL,k\in\cK}\in\cC(\rho)} \frac{R_{\ell k}(\rho)}{\log(\rho)}. \label{definition_dof}
\eeq
A network has $\bSig_{d}$ degrees of freedom if the sum capacity is expressed as $\bSig_{d}\log(\rho)\+o(\log(\rho))$.
This implies that the degrees of freedom $\bSig_d$ is equivalent to the total number of interference free signal
dimensions (i.e., the number of  effective single-input single-output (SISO) data streams that can be supported).

The degrees of freedom measure $\bSig_d$ in \eqref{definition_dof} ignores any fixed (or vanishing) quantities
in the achievable sum rate expression as $\rho$ increases.
Notice that the quantity $\bSig_d$ in \eqref{definition_dof} is characterized as a convergence of random variables
$\ls \frac{R_{\ell k}(\rho)}{\log(\rho)} \rs_{\ell\in\cL, k\in\cK}$ as $\rho\rightarrow\infty$.
The degrees of freedom results in \cite{Jafar2,Jafar1,Cadambe1,Cadambe2,Gou,Suh} show this convergence as almost sure (a.s.) sense.
When we refer the degrees of freedom in Section \ref{section3}, \ref{section4}, and \ref{section5}, that implies $\bSig_d$
characterized with instantaneous achievable rates $\{ R_{\ell k}(\rho) \}_{\ell\in\cL,k\in\cK}$.
While, when we explore the multiuser diversity in Section \ref{section6},
we need to distinguish between the \emph{instantaneous} degrees of freedom
and the \emph{average} degrees of freedom in order to capture the detailed difference in user scaling
laws. Notice that the former includes the mode of the convergence in random sequences
$\ls \frac{R_{\ell k}(\rho)}{\log(\rho)} \rs_{\ell\in\cL, k\in\cK}$ as $\rho, K \rightarrow \infty$,
while the later does not include detailed convergence argument.

In what follows, we will omit the $\rho$ attached to $W_{\ell k}(\rho)$ and $R_{\ell k}(\rho)$.
In addition, with an abuse of notation, $\by_m(t)$, $\bz_m(t)$, and $\bx_{\ell k}(t)$ in \eqref{channel_model1}
are simplified to $\by_m$, $\bz_m$, and $\bx_{\ell k}$.

\subsection{Downlink Channel Model}\label{subsection2.3}
The uplink scenario is converted to the downlink scenario by changing the role of the transmitter and receiver
and defining the reciprocal channel for the downlink as shown in \cite{Gomadam, Suh, Honig1} (i.e., uplink and downlink duality).
By $L$-cell and $K$-user MIMO downlink, we mean the network in which there are total $L$
transmitters and $K$ distributed receivers in each of cells.
In the downlink, we use the index $k\ell$ to correspond to user $k$ in cell $\ell$ for $k\in\cK$ and $\ell\in\cL$.

The received vector at user $k$ in cell $m$ is expressed by
\beq
\by_{km} =  \sum_{\ell=1}^L \bH_{km,\ell} \bx_{\ell} + \bn_{km} \label{channel_model_downlink}
\eeq
where $\by_{km}$ and $\bn_{km}$ are the $N \times 1$ received vector and additive white Gaussian noise vector
(distributed as $\cC\cN(\bzero, \bI_{N})$), respectively, at user $km$.
In \eqref{channel_model_downlink}, $\bH_{km, \ell}\in\C^{N\times M}$ denotes the channel
matrix from transmitter $\ell$ to user $km$.
The \emph{nondegenerate} channel condition, channel input power constraint, and encoding scheme
are similarly defined as in uplink channel model.
We will use this downlink  model in Section \ref{section6} to investigate the degrees of freedom
with multiuser diversity.

\section{Degrees of Freedom Outer Bound of the $L$-cell and $K$-user MIMO MAC} \label{section3}

\subsection{Degrees of Freedom Outer Bound} \label{subsection3.1}
Given the channel model in \eqref{channel_model1}, we now formulate the degrees of freedom outer bound for the $L$-cell and $K$-user MIMO MAC
when transmitter $\ell k$ has $M_{\ell k}$ antennas and receiver $\ell$ has $N_{\ell}$ antennas.
The following is the main result of this section.

\begin{theorem} \label{theorem_MC_MAC_outer_bound}
The total degrees of freedom of the $L$-cell and $K$-user MIMO MAC with $L>1$ and $K\geq 1$, whose channel matrices are nondegenerate,
is bounded by
\beq
\bSig_d \leq \min\lp \sum_{\ell\in\cL, k\in\cK}M_{\ell k},  \sum\limits_{\ell\in \cL}N_{\ell}, \eta(\cW)  \rp \label{theorem_MC_MAC_outer_bound_equation}
\eeq
where
\beq
\eta(\cW) = \frac{\sum\limits_{\ell\in\cL, k\in\cK}
\min\lp \sum\limits_{q\in\cK}M_{\ell q}\+\sum\limits_{p\in\cL\bsh \ell} M_{pk}, \sum\limits_{p\in\cL} N_p,
     \max\lp \sum\limits_{q\in\cK}M_{\ell q}, \sum\limits_{p\in\cL\bsh \ell} N_p \rp, \max\lp \sum\limits_{p\in\cL\bsh \ell} M_{pk}, N_{\ell} \rp  \rp}{K+L-1} \label{eta}
\eeq
with $\cL=\ls 1, \ldots, L \rs$, $\cK=\ls 1, \ldots, K \rs$, and $\cW=\ls W_{\ell k} \rs_{\ell\in \cL, k\in\cK}$.
\end{theorem}

\begin{proof}
The approach taken to derive the outer bound in \eqref{theorem_MC_MAC_outer_bound_equation} is to split the whole
message set $\cW=\ls W_{\ell k} \rs_{\ell\in\cL, k\in\cK}$ into subsets,
derive the outer bound associated with each of the subsets, and
combine all of the outer bounds to gain the total degrees of freedom outer bound.
In addition, we assume perfect channel knowledge of all links at all nodes.

Suppose we reduce the $L$-cell and $K$-user MIMO MAC to an $L$-cell heterogenous MIMO uplink channel
where the $L-1$ cells (among $L$ cells) constitute a $(L-1)$-user MIMO interference channel (IC) and the remaining single cell forms a $K$-user MIMO MAC.
We refer to this network as the $(1, L-1)$ \emph{MAC-IC uplink HetNet}.
Fig. \ref{Fig6} represents the $(1,2)$ \emph{MAC-IC uplink HetNet} composed of a single cell $2$-user MIMO MAC and $2$-user MIMO interference channel.
This $(1, L-1)$ MAC-IC uplink HetNet is formed from the $L$-cell and $K$-user MIMO MAC by eliminating messages in $\cW$
that do not constitute the information flow in the $(1,L-1)$ MAC-IC uplink HetNet channel.

Let the $\ell$th cell among $L$ cells is designated as the $K$-user MIMO MAC.
Then, the rest of the $L-1$ cells forms an $(L-1)$-user MIMO interference channel
by picking the $k$th user in each of the cells in $\cL\bsh\ell$,
i.e., the index set for the $L-1$ users is $\ls 1k, \ldots, (\ell-1) \ k, (\ell+1) \ k, \ldots Lk\rs$.
Message sets associated with the $K$-user MIMO MAC and $(L-1)$-user MIMO interference channel are then
given by $\ls W_{\ell q} \rs_{q\in \cK}$ and $\ls W_{pk} \rs_{p\in\cL\bsh \ell}$, respectively.
We define these two disjoint message sets as
\beq
\cW^{\ell k}= \ls W_{\ell q} \rs_{q\in\cK} \cup \ls W_{pk} \rs_{p\in\cL\bsh \ell}. \label{message_spliting}
\eeq
The degrees of freedom outer bound is first argued for each of the $LK$ sets $\ls \cW^{\ell k} \rs_{\ell\in\cL, k\in\cK}$, and
$LK$ outer bounds are combined by accounting the overlapped messages.

Assume perfect cooperations between $K$ users in cell $\ell$ and between
$L-1$ users and the corresponding $L-1$ receivers in the $(L-1)$-user MIMO interference channel.
Then, the $(1,L-1)$ MAC-IC uplink HetNet with $\cW^{\ell k}$ becomes a two-user interference channel with
transmit and receive antenna pairs $\lp \sum\limits_{q\in \cK}M_{\ell q}, N_{\ell} \rp$ for the first link and
$\lp  \sum\limits_{p\in\cL\bsh\ell}M_{p k},  \sum\limits_{p\in\cL\bsh\ell}N_{p} \rp$ for the second link.
It is well known that the spatial degrees of freedom of an $(M_1, N_1)$, $(M_2, N_2)$ two-user MIMO interference channel
is characterized as $\min( M_1+M_2, N_1+N_2, \max(M_1, N_2), \max(M_2, N_1) )$ \cite{Jafar2}.
Thus, the degrees of freedom outer bound associated with message set $\cW^{\ell k}$ is characterized by
\beq
\min\lp \sum\limits_{q\in\cK}M_{\ell q}\+\sum\limits_{p\in\cL\bsh \ell} M_{pk}, \sum\limits_{p\in\cL} N_p
,\max\lp \sum\limits_{q\in\cK}M_{\ell q}, \sum\limits_{p\in\cL\bsh \ell} N_p \rp, \max\lp \sum\limits_{p\in\cL\bsh \ell} M_{pk}, N_{\ell} \rp  \rp. \label{3.15}
\eeq
In the same manner, the outer bound associated with the message set $\cW^{\bar{\ell} \bar{k}}$ with $\bar{\ell} \neq \ell$ or
$\bar{k}\neq k$ is also determined by \eqref{3.15}.
Since there are total $KL$ message subsets and each message repeats $K+L-1$ times over $KL$ message subsets
(following from the splitting approach in \eqref{message_spliting}),
from \eqref{3.15} the total degrees of freedom associated with $\cW$ is bounded by
\beq
\bSig_d
\leq
\frac{\sum\limits_{\ell\in\cL, k\in\cK} \min\lp \sum\limits_{q\in\cK}M_{\ell q}\+\sum\limits_{p\in\cL\bsh \ell} M_{pk}, \sum\limits_{p\in\cL} N_p ,
      \max\lp \sum\limits_{q\in\cK}M_{\ell q}, \sum\limits_{p\in\cL\bsh \ell} N_p \rp, \max\lp \sum\limits_{p\in\cL\bsh \ell} M_{pk}, N_{\ell} \rp  \rp}
      {K+L-1}. \label{3.22}
\eeq

Meanwhile, a trivial bound is obtained by allowing perfect cooperation among $KL$ transmitters and full cooperation corresponding
$L$ receivers of the $L$-cell and $K$-user MIMO MAC as
\beq
\bSig_d \leq \min\lp \sum_{\ell\in\cL, k\in\cK}M_{\ell k},  \sum\limits_{\ell\in \cL}N_{\ell} \rp. \label{3.24}
\eeq
Combining two bounds in \eqref{3.22} and \eqref{3.24} yields the outer bound result in \eqref{theorem_MC_MAC_outer_bound_equation}.
\end{proof}

The characterized bound is general, in that it includes networks with $K\geq 1$ and $L>1$ for arbitrary numbers of
transmit and receive antennas.

The converse result in \eqref{theorem_MC_MAC_outer_bound_equation} can be further relaxed and simplified
by upper bounding $\eta(\cW)$ in \eqref{eta} as
\beq
\eta(\cW)\leq
\min\!\lp \! \sum_{\ell\in\cL, k\in\cK} \!\! M_{\ell k}, \frac{KL\cdot\sum\limits_{p\in\cL} \!N_p}{K+L-1},
     \frac{\sum\limits_{\ell\in\cL, k\in\cK} \! \min\!\lp\! \max\lp \sum\limits_{q\in\cK}M_{\ell q}, \sum\limits_{p\in\cL\bsh \ell} N_p \rp, \max\lp \sum\limits_{p\in\cL\bsh \ell} M_{pk}, N_{\ell} \rp \!\rp}
          {K+L-1} \!\rp \label{eta1}
\eeq
where in \eqref{eta1} the summation $\sum\limits_{\ell\in\cL, k\in\cK}$ is taken for operands inside of $\min(\cdot)$ in \eqref{eta} and
we use the facts that
\beq
\sum\limits_{\ell\in\cL, k\in\cK} \frac{\sum\limits_{q\in\cK}M_{\ell q}\+\sum\limits_{p\in\cL\bsh \ell} M_{pk}}{K+L-1} = \sum_{\ell\in\cL, k\in\cK}M_{\ell k} \nonumber
\eeq
and
\beq
\sum\limits_{\ell\in\cL, k\in\cK} \frac{\sum\limits_{p\in\cL} N_p}{K+L-1}= \frac{KL\sum\limits_{p\in\cL} N_p}{K+L-1}. \nonumber
\eeq
Since $\frac{KL}{K+L-1}\sum\limits_{p\in\cL} N_p\geq \sum\limits_{\ell\in\cL} N_{\ell}$ for $K,L\geq 1$,
combining the two bounds in \eqref{eta1} and \eqref{3.24} yields
\beq
\bSig_d\leq
\min\lp \sum_{\ell\in\cL, k\in\cK}  M_{\ell k}, \sum\limits_{\ell\in\cL} N_{\ell},
     \frac{\sum\limits_{\ell\in\cL, k\in\cK}\min\lp \max\lp \sum\limits_{q\in\cK}M_{\ell q}, \sum\limits_{p\in\cL\bsh \ell} N_p \rp, \max\lp \sum\limits_{p\in\cL\bsh \ell} M_{pk}, N_{\ell} \rp \rp}
          {K+L-1} \rp. \label{relaxed_bound}
\eeq
As mentioned earlier, our focus is mainly on an homogeneous antenna distribution.
The next corollary presents the required outer bound.
\begin{corollary} \label{corollary_homogenous_outerbound}
The total spatial degrees of freedom of the $L$-cell and $K$-user MIMO MAC with $M$ transmit antennas and $N$ receive antennas
is bounded by
\beq
\bSig_d \leq \min\lp KLM, LN, \frac{KL}{K+L-1}\max\lp KM, (L-1)N \rp, \frac{KL}{K+L-1}\max\lp (L-1)M, N \rp \rp. \label{corollary_homogenous_outerbound_equation}
\eeq
\end{corollary}
\begin{proof}
The bound can be obtained by substituting $M_{\ell k}=M_{\ell q}=M_{pk}=M$ and $N_{\ell}=N_p=N$ in \eqref{relaxed_bound} and taking
all the summations.
\end{proof}

\subsection{$(1,L-1)$ MAC-IC Uplink HetNet}
The characterized outer bound utilizes insight from a limit of the total degrees of freedom for an $L$-cell heterogeneous network, i.e.,
$(1,L-1)$ MAC-IC uplink HetNet. Denote $M_q$ and $N$ as the numbers of antennas at user $q$ and the base station in the $K$-user MIMO MAC, respectively,
and represent $M_p$ and $N_p$ as the number antennas at user $p$ and the corresponding receiver in the $(L-1)$-user MIMO interference channel, respectively.

\begin{corollary} \label{corollary_HetNet}
Denote $\bSig_{L-1,1}$ as the total degrees of freedom of the $(L-1,1)$ \emph{MAC-IC uplink HetNet}. Then,
\beq
\bSig_{L-1,1} \leq
\min\lp \sum\limits_{q=1}^{K} M_{q}\+\sum\limits_{p=1}^{L-1} M_{p}, \sum\limits_{p=1}^{L} N_{p} ,
        \max\lp \sum\limits_{q=1}^{K} M_{q}, \sum\limits_{p=1}^{L-1} N_p \rp, \max\lp \sum\limits_{p=1}^{L-1} M_{p}, N \rp  \rp
\eeq
\end{corollary}
\begin{proof}
Omit $\ell$ and $k$ attached to $M_{\ell q}$, $N_{\ell}$, and $M_{pk}$ in \eqref{3.15}.
Then, the formula in \eqref{3.15} verifies the corollary.
\end{proof}

Interestingly, the collocated $(L-1)$-user MIMO interference channel and single cell $K$-user MIMO MAC can be viewed as
a two-tier cell deployment where the network consists of $L-1$ femtocells (or picocells) each with a single user
and one macrocell with $K$ users.
Notice that in the two-tier networks, single user transmission at the lower-tier cell is shown to provide
significantly improved throughput and coverage than multiuser transmission \cite{Chandrasekhar}.


\subsection{Virtual MIMO Transmission vs. Selected and Shared Transmission}
Now we are interested in an equivalent channel model to the $L$-cell and $K$-user MIMO MAC.
Consider groups of $L$ distinct users among the $LK$ users (i.e., a total of $K$ user groups) such that the $k$th user
group is formed by grouping the $k$th user in each of the cells, i.e., the $k$th user group is the index set $\ls 1k, 2k, \ldots, Lk \rs$.
For example, Fig. \ref{Fig3a} shows the user grouping for the $L=3$ and $K=2$ MIMO MAC where the first user group
is represented as the index set $\{11,21,31\}$, and the second user group consists of indices $\{12,22,32\}$.
Then, the network is converted to a distributed $K\times L$ homogenous MIMO X channel (see Fig. \ref{Fig3b}).
Here, the equivalent channel of the $L$-cell and $K$-user MIMO MAC is referred to as
the \emph{distributed} $K\times L$ homogenous MIMO X channel because perfect cooperation
among users within each user group is not assumed\footnote{Notice that to meet the original
definition of the X channel in \cite{Maddah,Jafar2, Cadambe1}, the users within the $k$th user group must be perfectly connected,
i.e., in this case, the channel becomes a $K\times L$ MIMO X channel with $LM$ antennas at the transmitter and $N$ antennas at the receiver.
}.

The equivalency between the $L$-cell and $K$-user MIMO MAC and distributed $K\times L$ homogeneous MIMO X channel provides
an interesting insight into the following question:
When using spatial dimensions to transmit messages $\ls W_{\ell k} \rs_{\ell\in\cL, k\in\cK}$, is it better to employ
\emph{multiple distributed transmission} where transmitter $\ell k$, equipped with $M$ antennas, transmits its own message $W_{\ell k}$ or
to employ \emph{selected and shared transmission} where one transmitter, say $1k$ in the $k$th user group
$\ls 1k, 2k, \ldots, Lk \rs$, equipped with $M$ antennas, is selected and transmits all of the messages
$\ls W_{1k}, W_{2k}, \ldots, W_{Lk} \rs$ while other transmitters in the group keep quiet?
Given full CSI at all nodes, \emph{multiple distributed transmission} delivers messages $\ls W_{\ell k} \rs_{\ell \in\cL, k\in\cK}$
through distributed transmitters with the use of total $LKM$ dimensions (e.g., virtual MIMO transmission),
while \emph{selected and shared transmission} uses $KM$ dimensions with the use of partial message sharing through
the perfect links between transmitters.
We can show the later strategy is better in terms of the degrees of freedom than the former strategy for $L=2$ and $K=2$
(see Fig. \ref{Fig5} (a) and Fig. \ref{Fig5} (b)) as follows.

\begin{corollary}
Let $\bSig_{distTX}$ and $\bSig_{shrdTX}$ denote the total degrees of freedom of the \emph{multiple distributed transmission} and
\emph{selected and shared transmission}, respectively, when $L=2$ and $K=2$ with $M=N$. Then,
\beq
\bSig_{distTX} \leq \bSig_{shrdTX}. \nonumber
\eeq
\end{corollary}
\begin{proof}
Since the \emph{multiple distributed transmission} with $L=2$ and $K=2$ in Fig. \ref{Fig5} (a) is equivalent to
$2$-cell and $2$-user MIMO MAC, from Corollary \ref{corollary_homogenous_outerbound}
\beq
\bSig_{distTX} \d4&=&\d4 \bSig_{d} \nonumber \\
               \d4&\leq&\d4 \min\lp 4M, 2M, \frac{4}{3}\max(2M,M), \frac{4}{3}\max(M,M) \rp=\frac{4}{3}M. \nonumber
\eeq
The \emph{selected and shared transmission} through perfect link with $L=2$ and $K=2$
is the $2\times 2$ MIMO X channel with $M$ antennas at each node. Hence,
\beq
\bSig_{shrdTX}=\frac{4}{3}M \nonumber
\eeq
where the last equality follows from the optimal degrees of freedom result in \cite{Jafar1}
where the achievable scheme utilizes the simple zero forcing.
\end{proof}

In what follows, we will quote the results in  this section to characterize the optimal degrees of freedom
for $L$-cell and $K$-user MIMO MAC.

\section{Achieving the Optimal Degrees of Freedom} \label{section4}
In the homogenous $L$-cell and $K$-user MIMO MAC, independently encoded $\beta>0$ streams are transmitted as
$\bx_{mk} = \bT_{mk}\bs_{mk}$ from user $mk$ to base station $m$, where $\bs_{mk}\=\lS s_{mk,1} \ldots s_{mk,\beta} \rS^T$
is the $\beta\times 1$ symbol vector carrying message $W_{mk}$ and $\bT_{mk}\in\C^{M\times \beta}$ denotes a linear precoder
which will be chosen to provide interference free signal dimensions to user $mk$.
The $N$-dimensional signal received at base station $m$ is expressed as
\beq
\by_m\=\sum_{k=1}^{K}\!\bH_{m,mk}\bT_{mk}\bs_{mk}\+\sum_{\ell\neq m}^{L}\sum_{k=1}^{K}\bH_{m,\ell k}\bT_{\ell k}\bs_{\ell k}\+\bz_m.  \label{4.0}
\eeq
The achievable schemes must deal with $K(L\-1)\beta$ out-of-cell interference sources
and additionally $(K\-1)\beta$ inner cell interference sources.
This implies that the required spatial antenna dimensions $M$ and $N$ for the zero interference condition with constant channel
coefficients must be determined as a function of $K$, $L$, and $\beta$.

Our base line algorithm is to explore the feasibility of the linear schemes
utilizing the spatial dimensions under zero interference constraints.
Given \eqref{4.0}, our base line algorithm utilizes linear postprocessing matrix $\bP_m\in \C^{ K\beta \times N}$
at receiver $m$ to produce $\beta$ interference free dimensions for each of users.
The two-cell MIMO MAC scenario, which is instructive, is first considered, and a general multicell case is characterized later.

\subsection{Two-Cell MIMO MAC ($L=2$)} \label{section4.1}
The degrees of freedom outer bound in \eqref{corollary_homogenous_outerbound_equation} and zero forcing-based linear schemes
allow the following theorem to be proven.

\begin{theorem} \label{theorem_Two_cell_MAC_optimal_dof}
The two-cell and $K$-user MIMO MAC with the nondegenerate channels, where the
transmitter and receiver have $M\=K\beta$ and $N\=K\beta\+\beta$ or $M\=K\beta\+\beta$
and $N\=K\beta$ antennas, respectively, has the optimal degrees of freedom of $2K\beta$ where $\beta>0$ is a positive integer.
\end{theorem}

\subsubsection*{Converse of Theorem \ref{theorem_Two_cell_MAC_optimal_dof}} \label{outer_bound_MC_MAC_two_cell}
When $M=K\beta+\beta$ and $N=K\beta$, the outer bound in \eqref{corollary_homogenous_outerbound_equation} returns
\beq
\bSig_d \d4&\leq&\d4 \min\!\lp\! 2KM,2N, \frac{2K\max(\!KM,N\!)}{K\+1}, \frac{2K\max(\!M,N\!)}{K\+1} \!\rp  \nonumber \\
\d4&=&\d4 \min\!\lp\! 2K(K\+1)\beta, 2K\beta, \frac{2K^2 (K\+1)\beta}{K\+1}, 2K\beta \!\rp = 2K\beta \label{4.0.0}.
\eeq
When $M=K\beta$ and $N=K\beta+\beta$, we have
\beq
\bSig_d \leq \min\!\lp\! 2K^2\beta, 2(K+1)\beta, \frac{4K^3}{K\+1}, 2K\beta\!\rp = 2K\beta \label{4.0.1}.
\eeq
Combining two quantities in \eqref{4.0.0} and \eqref{4.0.1} verifies the converse. \hfill \qed

\subsubsection*{Achievability of Theorem \ref{theorem_Two_cell_MAC_optimal_dof}} \label{inner_bound_MC_MAC_two_cell}
The achievability is argued by showing that $\beta$ interference free dimensions per user are resolvable
at each of base stations.
For simplicity, we define $\barm$ as $\barm$$=$$\cL\bsh m$ where $\cL=\{1,2\}$ for two-cell case.

\subsubsection{$M=K\beta+\beta$ and $N=K\beta$}
When $M\= K\beta\+\beta$ and $N\=K\beta$, user $\barm k$ picks the precoding matrix
$\bT_{\barm k}$ such that
\beq
\span\lp \bT_{\barm k} \rp \subset null\lp \bH_{m,\barm k} \rp, \ k\in\cK. \label{4.2}
\eeq
Since $\bH_{m,\barm k}$$\in$$\C^{K\beta\times (K\beta\+\beta)}$ is drawn from an i.i.d. continuous distribution,
$\bT_{\barm k}$$\in$$\C^{M\times \beta}$ with $\rank(\bT_{\barm k})\=\beta$ can be found almost surely such that
$\bH_{m,m k}\bT_{\barm k}=\bzero$ for all $k\in \cK$.
In this way, user $\barm k$ precludes interference to base station $m$.
Applying precoders $\ls\bT_{\barm k}\rs_{k\in\cK, \barm\in\cL}$ designed by \eqref{4.2} to \eqref{4.0} yields
\beq
\by_m = \sum\limits_{k\in\cL}\bH_{m,mk}\bT_{mk}\bs_{mk} + \bz_m. \nonumber
\eeq

The decodability of $K\beta$ dimensions from $\by_m$ requires
\beq
\bG_m=\lS \bH_{m,m1}\bT_{m1} \ \cdots \ \bH_{m,mK}\bT_{mK} \rS \in \C^{K\beta \times K\beta}  \label{4.3}
\eeq
to be a full rank.
Since $\bT_{mk}$ in \eqref{4.2} is based on $\bH_{\barm,m k}$, $\bT_{mk}$ is mutually independent of $\bH_{m,mk}$.
Then,
by Lemma \ref{lemma_linear_independent_columns} in Appendix \ref{appendix_linear_independent_columns},
$\bH_{m,mk}\bT_{mk}\in\C^{K\beta\times \beta}$ is a full rank and spans a $\beta$-dimensional subspace with probability one.
Since $\ls \bH_{m,mk}\bT_{mk} \rs_{k\in\cK}$ are independently realized by continuous distributions and
each $\bH_{m,mk}\bT_{mk}$ spans $\beta$-dimensional subspace, the aggregated channel $\bG_m\in\C^{K\beta\times K\beta}$
spans $K\beta$-dimensional space almost surely.
This ensures achievability of $2K\beta$ degrees of freedom when $M=K\beta+1$ and $N=K\beta$.

\subsubsection{$M=K\beta$ and $N=K\beta+\beta$}
When $M=K\beta$ and $N=K\beta+\beta$, an achievable scheme employs the postprocessing matrix
$\bP_m\in\C^{K\beta\times (K\beta\+\beta)}$ designed at base station $m$.

Suppose a set of matrices $\{ \lS \bH_{m,\barm k} \ \bN_{m,\barm k} \rS \}_{k\in\cK}$ where matrix
$\lS \bH_{m,\barm k} \ \bN_{m,\barm k} \rS\in \C^{(K\beta+\beta)\times(K\beta+\beta)}$ is formed by concatenating
two matrices $\bH_{m,\barm k}\in \C^{(K\beta+\beta)\times K\beta}$ and $\bN_{m,\barm k}\in\C^{(K\beta+\beta)\times \beta}$
such that $\lS \bH_{m,\barm k} \ \bN_{m,\barm k} \rS$ is full rank matrix for $k\in\cK$, i.e., $\bN_{m,\barm k}^*\bH_{m,\barm k}=\bzero$.
Then, $\bP_m$$\in$$\C^{K\beta\times (K\beta+\beta)}$ is designed such that
\beq
\span\lp \bP_m^* \rp = \span\lp \lS \bN_{m,\barm 1} \  \bN_{m,\barm 2} \ \cdots \bN_{m,\barm K} \rS \rp,  \label{4.6.0}
\eeq
i.e., the column subspace of $\bP_m^*$ spans the same column subspace as
$\lS \bN_{m,\barm 1} \  \bN_{m,\barm 2} \ \cdots \bN_{m,\barm K} \rS\in\C^{(K\beta+\beta)\times K\beta}$.
By \eqref{4.6.0}, $\bP_m$ is constructed by
\beq
\bP_m \= \bPi \lS \bN_{m,\barm 1} \  \bN_{m,\barm 2} \ \cdots \bN_{m,\barm K} \rS^*, \ m\in\cL \label{4.6}
\eeq
where $\bPi\in\C^{K\beta\times K\beta}$ is any full rank matrix.
Notice the construction in \eqref{4.6} with $\{\bN_{m,\barm k}\}_{k\in\cK}$ always ensures $\rank(\bP_m)\=K\beta$ and
\beq
\dim\lp null\lp \bP_{m}\bH_{m,\barm k} \rp \rp= \beta \label{4.6.1}
\eeq
for all $k\in\cK$.

Given $\ls \bP_m \rs_{m\in\cL}$ in \eqref{4.6}, we find the precoder $\bT_{\barm k}\in\C^{K\beta \times \beta}$ under the zero out-of-cell
interference constraint such that
\beq
\span\lp \bT_{\barm k} \rp\subset null\lp \bP_{m}\bH_{m,\barm k} \rp, \ k\in\cK, \ \barm\in\cL, \nonumber
\eeq
where such $\bT_{\barm k}$ with $\rank\lp \bT_{\barm k} \rp=\beta$ exists almost surely because of \eqref{4.6.1}.
Then, the projected channel output at the base station $m$ is given by
\beq
\bP_m\by_m \!\=\! \sum_{k=1}^{K}\bP_m\bH_{m,mk}\!\bT_{mk}\bs_{mk}\+\bP_m\bz_m \!\=\! \bP_m\bG_m \tbs_m \+ \tbz_m \label{4.8}
\eeq
where $\bG_m=[ \bH_{m,m1}\!\bT_{m1}  \cdots \bH_{m,mK}\!\bT_{mK} ]\in\C^{(K\beta+\beta)\times K\beta}$,
$\tbz_m=\bP_m\bz_m$, and $\tbs_m=[\bs_{m1}^T \cdots \bs_{mK}^T]^T$.
For decodability, we need to check that $\bP_m\bG_m$ has linearly independent columns.
Analogous to \eqref{4.3}, $\bG_m$ in \eqref{4.8} spans a $K\beta$-dimensional subspace almost surely.
Note that $\bP_m$ in \eqref{4.6} and $\bG_m$ are based on a continuous distribution and are mutually independent.
Thus,
$\Pr\big( \det\big(  \bP_m\bG_m \big)\=0 \big)\=0$ (by Lemma \ref{lemma_linear_independent_columns} in Appendix \ref{appendix_linear_independent_columns})
implying the decodability of $K\beta$ interference free streams per cell.
\hfill \qed

When $M=K\beta$ and $N=K\beta+\beta$, the achievable scheme aligns the null spaces of the out-of-interference channel $\{ \bH_{m,\barm k}^* \}_{k\in\cK}$
to the row subspace of $\bP_m$, which is referred to as \emph{null space interference alignment}.
In the null space interference alignment, the post processing matrix $\bP_m$ compresses $K\beta$-dimensional
out-of-cell interference channels to $(K-1)\beta$-dimensional signal subspace because the $\beta$-dimensional row subspace of
$\bP_m$ always lies in $null( \bH_{m,\barm k}^* )$ for all $k\in\cK$.
In fact, since the condition in \eqref{4.6.1} describes the required condition about the right matrix null space of $\bP_m\bH_{m,\barm k}$,
omitting the full rank matrix $\bPi\in\C^{K\beta\times K\beta}$ on the left side of $\bP_m$ does not change the dimension
condition in \eqref{4.6.1}, i.e.,
\beq
\dim\lp null\lp \bPi^{-1}\bP_{m}\bH_{m,\barm k} \rp \rp=\dim\lp null\lp \bP_{m}\bH_{m,\barm k} \rp \rp= \beta, \  k\in\cK. \label{4.7}
\eeq

We have discussed the achievability of the optimal degrees of freedom for the two cell case by using
transmit zero forcing (with $M=K\beta+\beta$ and $N=K\beta$) and null space interference alignment
(with $M=K\beta$ and $N=K\beta+\beta$) for arbitrary $K>0$ and $\beta>0$.
As will be seen in Section \ref{section5}, the basic idea of the null space interference alignment
can be generalized for $L\geq 2$ with $N>M$. The generalized scheme does not necessarily
achieve the optimal degrees of freedom, but it resolves achievable $\beta>0$
interference free dimensions for each of users with various antenna dimensional conditions.

\subsection{Multicell MIMO MAC ($L\geq2$)} \label{section4.2}

In the uplink, the scenario of $N>M$ is realistic because the system dimension at
the user side is often limited.
In this scenario, one of the extreme choices for $M$ and $N$ is when the user has $\beta$ antennas for
$\beta$ stream multiplexing, i.e., $M=\beta$, and interference cancellation is mainly accomplished
at the base station. As will be seen in the next theorem, employing the minimum number of transmit
antennas generally achieves the optimal degrees of freedom for $L$-cell and $K$-user MIMO MAC.

\begin{theorem} \label{theorem_RX_ZF_optimal_dof}
Given $M=\beta$ transmit antennas and $N=LK\beta$ receive antennas,
the $L$-cell and $K$-user MIMO MAC with nondegenerate channel matrices
has the optimal degrees of freedom of $LK\beta$.
\end{theorem}
\begin{proof}
See Appendix \ref{appendix_theorem_RX_ZF_optimal_dof}.
\end{proof}

The inner bound of the theorem is shown by using simple receive zero forcing.
The theorem suggests that given full CSI at the base stations, other than allowing some level of coordinated transmit
and receive filtering, employing base station-centric interference nulling scheme is potentially simple and reliable
in the high SNR regime in the multicell multiuser MIMO uplink scenario (some of which can be practically achieved in small cell scenarios).

Analogous to \cite{Jafar2,Jafar1}, Theorem \ref{theorem_Two_cell_MAC_optimal_dof} and Theorem \ref{theorem_RX_ZF_optimal_dof}
show that the simple zero forcing is indeed optimal in terms of the achievable degrees of freedom for $L$-cell and $K$-user MIMO MAC.

\section{General Framework for the Null Space Interference Alignment} \label{section5}

Complete characterization of the optimal spatial degrees of freedom with constant channel coefficients
for the $L$-cell and $K$-user MIMO networks is still unknown and
often overconstrained.
However, this difficulty does not preclude the existence of a general linear scheme that resolves $\beta>0$
interference free dimensions per user.
In this section, the basic idea of the null space interference alignment (with $N>M$) in Section
\ref{inner_bound_MC_MAC_two_cell} is extended to a general framework.

Throughout the section, we will use following two definitions to measure the size of
overlapping of the out-of-cell interference null space.

Suppose there are $K$ i.i.d. full rank matrices (i.e., nondegenerate)
$\big\{ \lS \bA_{k} \ \bB_{k} \rS \big\}_{k\in\cK}$, $\cK=\ls 1,2,\ldots, K \rs$, where $\lS \bA_k \ \bB_k \rS$
is square and invertible with $\bA_{k}\in\C^{n\times m}$ and $\bB_{k}\in\C^{n\times (n-m)}$ ($n>m$).

\begin{definition} \label{definition1}
A set $\big\{ \bA_{k} \big\}_{k\in\cK}$ is referred to as having a null space with \emph{geometric multiplicity} $\gam$,
if all $\gam$-tuple combinations of the matrices $\{ \bB_{\pi_1 }, \ldots, \bB_{\pi_{\gam} } \}$
with $\{\pi_i\}_{i=1}^{\gam}\subset \cK$, $\pi_i\neq\pi_j$ if $i\neq j$, have nonempty intersection, i.e.,
\beq
\bigcap\limits_{i=1}^{\gam} ran(\bB_{\pi_{i} })\neq \phi \nonumber
\eeq
and at the same time $\gam$ is the \emph{maximum} possible value.
\end{definition}

\begin{definition} \label{definition2}
Given $\gam\geq 1$ in Definition \ref{definition1}, the intersection null space of $\ls \bA_{k} \rs_{k\in\cK}$
is referred to as having \emph{algebraic multiplicity} $\mu$ if
\beq
\mu\= \dim\lp \bigcap\limits_{i=1}^{\gam} ran(\bB_{\pi_i }) \rp.  \nonumber
\eeq
\end{definition}

The quantities $\gam$ and $\mu$ in Definition \ref{definition1} and \ref{definition2}, respectively,
can be formulated as in the following lemma that elucidates the linear algebraic relation between $\gam$ and $\mu$.

\begin{theorem} \label{theorem_geometric_multiplicity}
Given a set of nondegenerate full rank matrices $\ls \lS \bA_{k} \ \bB_k \rS \rs_{k\in\cK}$ with $\cK=\ls 1, \ldots, K \rs$
where $\bA_{k}\in\C^{n\times m}$ ($n>m$) and $\bB_{k}\in\C^{n\times (n-m)}$,
respectively, the geometric multiplicity $\gam$ of $\ls \bA_{k} \rs_{k\in\cK}$
is characterized by
\beq
\gam\=\min\lp \left\lceil \frac{n-m}{m} \right\rceil, K \rp \nonumber
\eeq
and the algebraic multiplicity $\mu$ ($1\leq \mu\leq m$) satisfies
\beq
\mu=n-\gam m. \nonumber
\eeq
\end{theorem}
\begin{proof}
See Appendix \ref{appendix_theorem_geometric_multiplicity}.
\end{proof}

The scheme requires different pairs of $M$ and $N$
depending on the size of the overlapped interference null space dimension in order to preserve $\beta$ interference free dimensions per user.
We elaborate the framework for the two-cell case and the scheme is directly extended to the $L>2$ cell case,
which is provided in Appendix \ref{appendix_NIA_L_geq_2}.

For the two-cell case, given $K$ out-of-cell interference channels $\big\{ \bH_{m,\barm k} \big\}_{k\in\cK}$ with $\bH_{m,\barm k}\in\C^{N\times M}$
and corresponding null space $\ls \bN_{m,\barm k}\rs_{k\in\cK}$ where $\bN_{m,\barm k}\in\C^{N\times (N-M)}$ such that
$\lS \bH_{m, \barm k} \ \bN_{m,\barm k} \rS$ is full rank,
$\gam$ of $\big\{ \bH_{m,\barm k} \big\}_{k\in\cK}$ is given by
\beq
\gam\=\min\lp \left\lceil \frac{N-M}{M} \right\rceil, K \rp \nonumber
\eeq
by Theorem \ref{theorem_geometric_multiplicity}.
Since $N>M$, $\gam$ is bound by $1\leq \gam \leq K$.
The generalized null space interference alignment scheme is described by determining required $M$ and $N$ for a given value of $\gam$ ($1\leq \gam\leq K$)
such that the scheme can resolve $\beta$ interference free dimensions per users.

Under the zero out-of-cell interference constraint, given $\bP_m\in\C^{K\beta\times N}$, the precoder $\bT_{\barm k}\in\C^{M\times \beta}$ must lie in the
null space of $\bP_m\bH_{m,\barm k}$, i.e., $\span(\bT_{\barm k})\subset null(\bP_m\bH_{m,\barm k})$ for $k\in\cK$.
The condition $\span(\bT_{\barm k})\subset null(\bP_m\bH_{m,\barm k})$
is accomplished if
\beq
\dim\lp null(\bP_m\bH_{m,\barm k}) \rp\geq \beta, \ k \in\cK. \label{5.1.0}
\eeq
With the equality
$\dim\lp null\lp \bP_m\bH_{m,\barm k} \rp \rp = M - \rank\lp \bP_m\bH_{m,\barm k} \rp$ for $k\in \cK$,
we have
\beq
M \geq \rank\lp \bP_m\bH_{m,\barm k} \rp + \beta, \ k \in \cK. \label{5.3}
\eeq
The formula \eqref{5.3} implies that in order to accomplish the zero out-of-cell interference, we need
$\rank\lp \bP_m\bH_{m,\barm k} \rp < M$, $k \in \cK$
with $N>M$, while $\rank\lp \bP_m\bH_{m,\barm k} \rp \leq \min(K\beta, M)$, implying
\beq
\rank\lp \bP_m\bH_{m,\barm k} \rp \leq K\beta. \label{5.4}
\eeq

Given the $\gam$, the feasible $\bP_m\in \C^{K\beta\times N}$ and the antennas dimensions $N$ and $M$
that satisfies \eqref{5.3} can be designed by assigning $\gam$-overlapped intersection null spaces
of some groups of out-of-cell interference channels to the row subspace of $\bP_m$.

\textbf{Step1}: Let us define $k$th $\gam$-tuple index set as $\Pi_{k}=\ls \pi_i \rs_{i=k}^{\gam\+k\-1}$ for $k\in\cK$
with
\beq
\pi_i\=\lp (i-1) \ mod \ K  \rp\+1. \label{5.4.1}
\eeq
For instance, when $\gam=2$, $K=3$, and $L=2$, index group
$\ls \Pi_k \rs_{k=1}^3$ is composed of $\Pi_1=\ls 1,2 \rs$, $\Pi_2=\ls 2,3 \rs$, and $\Pi_3=\ls 3,1 \rs$.
The defined index group $\ls \Pi_k \rs_{k=1}^K$
ensures that every index in $\cK$ appears $\gam$ times throughout $K$ distinct sets.

\textbf{Step2}: Define the intersection null space associated with channel indices in $\Pi_k$ as
$\bN_{m,\barm}^{(k)}\in\C^{N\times \mu}$, i.e.,
\beq
\span\lp \bN_{m,\barm}^{(k)} \rp \subset \bigcap_{i=k}^{\gam+k-1} ran\lp \bN_{m, \barm \pi_i} \rp. \nonumber
\eeq
For $\{ \bH_{m,\barm i} \}_{i\in\Pi_k}$, the $\mu$-dimensional intersection null space $\bN_{m,\barm}^{(k)}$ is
efficiently found by using the iterative formula in \eqref{ap3.2} in Appendix \ref{appendix_theorem_geometric_multiplicity}.

\textbf{Step3}: When $1\leq \gam \leq K-1$, $\bN_{m,\barm}^{(k)}$ is found such that $\mu=\beta$ and
the row subspace of $\bP_m\in\C^{K\beta\times N}$ is constructed by
\beq
\bP_m \= \bPi \lS \bN_{m,\barm}^{(1)} \  \bN_{m,\barm}^{(2)} \ \cdots \bN_{m,\barm}^{(K)} \rS^* \label{5.8}
\eeq
where $\bPi\in\C^{K\beta\times K\beta}$ is a full rank matrix.
From Theorem \ref{theorem_geometric_multiplicity}, the existence of $\bN_{m,\barm}^{(k)}$ with $\mu=\beta$
is guaranteed if $N\=\gam M \+ \beta$.
When $\gam=K$, there exists only one intersection null space $\bN_{m,\barm}^{(1)}$ such that
$\span( \bN_{m,\barm}^{(1)} ) \subset \bigcap\limits_{k=1}^{K} ran\lp \bN_{m,\barm k} \rp$. In this case,
$\mu$ of $\bN_{m,\barm}^{(1)}$ is set to $\mu=K\beta$ and
\beq
\bP_m =  \bPi \bN_{m,\barm}^{(1)^{\scriptstyle *}}.  \label{5.8.1}
\eeq
The result in \eqref{5.8.1} is possible when $N=\gam M+K\beta$.

\textbf{Step4}: Given $N$ formulated in \textbf{Step3}, we now formulate the required dimension $M$.
The $\bP_m$ in \eqref{5.8} and \eqref{5.8.1} always contains $\gam\beta$-dimensional subspace
that is lying in the null space of $\bH_{m,\barm k}$ for all $k\in\cK$.
Thus, the projected out-of-cell interference channels $\ls \bP_m\bH_{m,\barm k} \rs_{k\in\cK}$ satisfies
\beq
\rank\lp \bP_m\bH_{m,\barm k} \rp = (K-\gam)\beta, \ k\in\cK. \label{5.8.2}
\eeq
Plugging \eqref{5.8.2} in \eqref{5.3}, the $M$ ensuring the zero out-of-cell interference constraint in \eqref{5.1.0}
yields
\beq
M = (K-\gam)\beta + \beta. \label{5.9}
\eeq

When $L>2$, the generalized null space interference alignment is presented in Appendix \ref{appendix_NIA_L_geq_2} which
utilizes channel aggregation.
The same decodability argument used in Section \ref{section4.1} can be applied for $L\geq 2$.
To avoid repetition we omit this part.

Now, given $\gam$ and $\beta$, the required $M$ for $L\geq 2$ is
\beq
M = (L-1)(K-\gam)\beta + \beta. \label{5.12.1}
\eeq
Then, the dimension $N$ to resolve $\beta$ interference free dimensions is given by
\beq
N = (L-1)\gam M + \beta \ \ \text{if} \ \ 1\leq \gam \leq K-1 \label{5.13}
\eeq
and
\beq
N = (L-1)\gam M + K\beta \ \ \text{if} \ \ \gam=K. \label{5.14}
\eeq


It can now be observed that the developed generalized framework includes the achievable schemes in
Theorem \ref{theorem_Two_cell_MAC_optimal_dof} and Theorem \ref{theorem_RX_ZF_optimal_dof}, i.e.,
when $\gam=1$, the generalized null space interference alignment attains the optimal degrees of freedom
for two-cell case and when $\gam=K$, the scheme shows the optimal degrees of freedom for $L\geq 2$.
For $2\leq \gam \leq K-1$, it does not necessarily achieve
the optimal degrees of freedom, rather it provides $\beta$ interference-free dimensions per user, i.e.,
it provides a total $KL\beta$ degrees of freedom.

Recently, a necessary condition for a linear achievable scheme providing one interference free
dimension per user (i.e., $\beta=1$) for $L$-cell and $K$-user MIMO network is characterized as \cite{Honig1}
\beq
M+N \geq LK+1. \label{5.15}
\eeq
This condition indicates that no linear scheme can provide even one interference free dimension per user,
if $M+N<LK+1$. In addition, the crucial metric $M+N$ in \eqref{5.15} measures the redundancy
in $M$ and $N$ to provide the $\beta=1$ interference free dimension per user.

\begin{remark}
Generalized null space interference alignment with $\beta=1$ always satisfies the necessary condition
$M+N\geq LK+1$. Moreover, the linear schemes in Theorem \ref{theorem_Two_cell_MAC_optimal_dof} and Theorem \ref{theorem_RX_ZF_optimal_dof}
achieve the optimal degrees of freedom with the minimum possible $M+N=LK+1$.
\end{remark}
\section{Leveraging Multiuser Diversity for $L$-cell Downlink MIMO Interference Channel} \label{section6}
We have argued the optimal spatial degrees of freedom and the generalized null space interference alignment scheme
with constant channel coefficients.
Allocating spatial resources across multiple users in the network is another dimension that has the
potential to provide additional spatial degrees of freedom with only a small amount of CSI feedback.

In this section, the degrees of freedom of the $L$-cell single-input
multiple-output (SIMO) downlink MIMO system by exploiting multiuser diversity is studied.
Thus, we consider the downlink channel model in \eqref{channel_model_downlink}.
We are particularly interested in a downlink receive beamforming system using $\beta=1$ stream transmission.

We look at an example where each transmitter has $M=1$ antennas and each receiver is equipped with $N=L-1$ antennas.
There is a total of $K$ users in each cell.
In order to exploit multiuser diversity, the user having the best channel is selected in the cell.
Notice that after the user selection, the network is reduced to an $L$-cell SIMO interference channel.
We first introduce the user selection strategy and characterize the instantaneous degrees of freedom and
average degrees of freedom as introduced in Section \ref{subsection2.2} and \ref{subsection2.3}.

\subsection{User Scheduling Framework}
Initially, $L$ basestations simultaneously transmit training symbols
$s_1, \ldots, s_L$ to all users in the network where $s_{\ell}\in\C^{1\times 1}$.
Then, the channel output vector at user $km$ is expressed by
\beq
\by_{km} = \bh_{km,m} s_{m} + \sum_{\ell\neq m}^{L} \bh_{km, \ell} s_{\ell} + \bn_{km} \label{6.1}
\eeq
where $\by_{km}$ and $\bn_{km}$ are the $(L-1) \times 1$ received vector and noise vector.

We assume that channel vectors in $\ls \bh_{km,\ell} \rs_{\ell,m\in\cL, k\in\cK}$
are mutually independent and realized so that each entry of $\bh_{km,\ell}$ is an i.i.d. zero mean
complex Gaussian random variable with unit variance, i.e., $\cC\cN(\bzero, \bI_{L-1})$.
The training symbol (or data symbol after the training phase) satisfies the average power constraint $E[ | s_m |^2 ]= \rho$.
The symbols are independently generated with $E\lS s_m s_{\ell}^* \rS=\rho$ for $m=\ell$ and zero otherwise.

The addressed user scheduling scheme does not assume global channel knowledge at all nodes;
in contract, user $km$ only has knowledge about its own channel $\bh_{km,m}$ and the
covariance matrix of the out-of-cell interference defined as
\beq
E\!\lS\! \sum\limits_{\ell\neq m}^{L} \bh_{km, \ell} s_{\ell} \lp \!\sum\limits_{\ell\neq m}^{L} \bh_{km, \ell} s_{\ell} \!\rp^*  \rS
\=\rho\sum\limits_{\ell\neq m}^{L} \bh_{km, \ell} \bh_{km,\ell}^*. \label{2.1.1}
\eeq
Thus, the scheme only requires local CSI, which significantly decreases the amount of CSI compared to
conventional interference alignment \cite{Jafar2, Jafar1, Cadambe1, Cadambe2, Gou,Suh}.

Denote the out-of-cell interference covariance matrix at user $km$  (i.e., the matrix in  \eqref{2.1.1})
as $\rho\bW_{km}$ meaning that $\rho\bW_{km}=\rho\sum\limits_{\ell\neq m}^{L} \bh_{km, \ell} \bh_{km,\ell}^*$.
Then, user $km$ selects a receive beamforming vector $\bp_{km}\in\C^{(L-1) \times 1}$
to maximize the signal to noise plus interference ratio (SINR) according to
\beq
\bp_{km} = \argmax_{\bp \in \C^{(L-1)\times 1}} \frac{\rho \la \bp^*\bh_{km,m} \ra^2}{\lA \bp \rA_2^2+ \rho\bp^*\bW_{ km}\bp }. \label{2.1.2}
\eeq
The solution to \eqref{2.1.2} is $\bp_{km}=\bv_{max,km}$ where $\bv_{max,km}$ is the
eigenvector associated with the largest eigenvalue of $\lp \bI_N + \rho \bW_{km} \rp^{-1}\rho\bh_{km,m}\bh_{km,m}^*$ meaning that
\beq
\lambda_{max,km} &=& \lambda_{max}\lp \lp \bI_N + \rho\bW_{km} \rp^{-1} \rho\bh_{km,m}\bh_{km,m}^* \rp \nonumber \\
&=& \frac{\rho \la \bp_{km}^*\bh_{km,m} \ra^2}{\lA \bp_{km} \rA_2^2+ \rho\bp_{km}^*\bW_{ km}\bp_{km} } \label{2.1.3}
\eeq
where $\lambda_{max}(\bA)$ returns the dominant eigenvalue of matrix $\bA$.

Users associated with transmitter $m$ feed back $\{ \lambda_{max,km} \}_{k\in\cK}$ through the
feedback link to transmitter $m$. Then, transmitter $m$ selects the best user such that
\beq
\hk m = \argmax_{k\in\cK} \lambda_{max,km}. \label{2.1.4}
\eeq

After the user selection, data symbols are transmitted to serve the selected $L$ users $\{ \hk m \}_{m\in\cL}$
from each base station in a cell.
Overall, the system reduces to an $L$-cell SIMO interference channel.

Passing the received signal vector at the selected user $\hk m$ through
the receive processing filter $\bp_{\hk m}$ yields
\beq
\bp_{\hk m}^* \by_{\hk m} \= \bp_{\hk m}^*\bh_{\hk m,m} s_{m} \+ \sum_{\ell\neq m}^{L} \bp_{\hk m}^*\bh_{\hk m, \ell} s_{\ell}
\+ \bp_{\hk m}^*\bn_{\hk m}, \label{2.1.5}
\eeq
and the instantaneous rate at user $\hk m$ is written as
\beq
{{R}}_{\hk m}(\rho) \!\=   \log\!\lp\!\! 1\+ \frac{\rho \la \bp_{\hk m}^*\bh_{\hk m, m} \ra^2}
{\lA \bp_{\hk m} \rA_2^2+\rho  \bp_{\hk m}^* \bW_{\hk m}\bp_{\hk m} }\!\rp. \label{2.1.6}
\eeq
Notice that
\beq
{{R}}_{\hk m}(\rho) = \max_{k\in\cK} {{R}}_{k m}(\rho).  \label{2.1.6.1}
\eeq

\subsection{Instantaneous Degrees of Freedom Analysis} \label{section6.1}
The approach taken to analyze the instantaneous degrees of freedom is to derive
a tractable inner bound and outer bound of the instantaneous degrees of freedom and
show that two bounds converge to the same quantity.
For this purpose, we first consider the inner bound scheme.

Given $(L-1)$-dimensional channel output vector,
user $km$ of the inner bound scheme selects receive processing vector $\tbp_{km}\in\C^{(L-1)\times 1}$
only to minimize the out-of-cell interference power such that
\beq
\tbp_{km}=\argmin_{\bp\in\C^{(L-1)\times 1}} \bp^* \bW_{km} \bp. \label{2.1.8}
\eeq
The minimizer in \eqref{2.1.8} is $\tbp_{km}=\bu_{min,km}$ where $\bu_{min,km}$ is the eigenvector
associated with the smallest eigenvalue of $\bW_{km}$, i.e.,
\beq
\sigma_{km}=\lambda_{min}\lp \bW_{km} \rp. \label{2.1.8.1}
\eeq
Users registered to transmitter $m$ feed back interference statistics $\ls \sigma_{km} \rs_{k\in\cK}$
through the feedback link to transmitter $m$.
Then, transmitter $m$ picks the best user such that
\beq
\hk m = \argmin_{k\in \cK} \sigma_{km} \label{2.1.9}
\eeq
where the scheduler in \eqref{2.1.9} is namely the minimum interference power scheduler.
After  post processing with $\tbp_{\hk m}$ in \eqref{2.1.8} at the receiver,
the achievable rate of the inner bound scheme is
\beq
{\tilde{R}}_{\hk m}(\rho) \!\=   \log\!\lp\!\! 1\+ \frac{\rho \la \tbp_{\hk m}^*\bh_{\hk m, m} \ra^2}
{\lA \tbp_{\hk m} \rA_2^2+\rho \tbp_{\hk m}^*\bW_{\hk m} \tbp_{\hk m} } \!\rp. \label{2.1.9.1}
\eeq

Obviously, the sum rate $\sum\limits_{m=1}^{L}\tilde{R}_{\hk m}(\rho)$ obtained
by the inner bound scheme is a lower bound of $\sum\limits_{m=1}^{L}{R}_{\hk m}(\rho)$
in \eqref{2.1.6} which is based on the maximum SINR scheduling in \eqref{2.1.4}.
The following lemma establishes the convergence law for the interference power in \eqref{2.1.8.1} which will
play a key role for showing the main result of this section.

\begin{lemma}\label{lemma_mean_square_conv}
If $\rho, K\rightarrow \infty$ while maintaining $K \varpropto \rho^a$ with $a>1$ and $a\in\R$, then
\beq
\rho \tbp_{\hk m}^* \bW_{\hk m} \tbp_{\hk m} = \rho\sigma_{\hk m}
\stackrel{\substack{m.s \\ a.s.}}{\longrightarrow}
0 \label{lemma_mean_square_conv_equation}
\eeq
in mean-square (m.s.) and almost sure (a.s.) sense.
\end{lemma}
\begin{proof}
First, notice that random variable $\min\limits_{k\in\cK} \sigma_{km}$ in \eqref{2.1.9}
is the minimum order statistic of i.i.d. $K$ minimum eigenvalues
of Wishart matrices $\bW_{1 m}, \ldots,  \bW_{K m}$ where
$\bW_{k m}= \bY_{km}\bY_{km}^*$ with $(L\-1)\times (L\-1)$ dimensional
$\bY_{km}=\lS \bh_{k m, 1} \cdots \bh_{k m, m\-1} \ \bh_{k m, m\+1} \cdots \bh_{k m, L}\rS$.
It was shown in \cite{Edelman} the probability density function (PDF) of the minimum eigenvalue
of Wishart matrix with $(L\-1)\times (L\-1)$ dimensional $\bY_{km}$ is given by $f(\sigma) = (L-1) e^{-(L-1)\sigma}$.
Thus, the PDF of $\rho\sigma_{km}$ is
\beq
f(\rho\sigma) = \frac{L-1}{\rho} e^{-\frac{L-1}{\rho}\sigma} \label{2.1.10}.
\eeq
From \eqref{2.1.10}, the complementary cumulative distribution function (CCDF) of $\rho\sigma_{k m}$ is derived as
$\Pr\lp \rho\sigma > x \rp \= e^{-\frac{L-1}{\rho}x}$.
Then, CCDF of $\rho\sigma_{\hk m}$ is
\beq
\Pr\!\lp \rho\sigma_{\hk m} \!>\! x \rp \=\lp \Pr\!\lp \rho\sigma > x \rp \rp^K \=  e^{-\frac{(L-1)K}{\rho}x}.\label{2.1.11}
\eeq
We first show the almost sure (a.s.) convergence and the argument for the mean-square (m.s.) convergence follows.

\subsubsection{Almost Sure Convergence}
For $\forall \ep >0$, as $\rho,K \rightarrow \infty$ in such a way that $K\propto \rho^a$ with $a>1$,
we have from \eqref{2.1.11}
\beq
\d4\Pr\!\lp \lim_{\rho,K \!\rightarrow \infty}\! \rho\sigma_{\hk m} \!> \ep \!\rp
\!\d4&=&\d4\! \lim_{\rho,K\rightarrow \infty}\!e^{-\frac{(L-1)K}{\rho}\ep} \nonumber \\
\!\d4&=&\d4\! \lim_{\rho,K\rightarrow \infty} e^{-(L-1)\rho^{a-1}\ep} \= 0. \nonumber
\eeq
Since this holds for arbitrarily small $\ep>0$, this implies
\beq
\Pr\lp \lim_{\rho,K\rightarrow \infty} \rho\sigma_{\hk m} \= 0 \rp
\= 1 \- \lim_{\ep\rightarrow 0}\Pr\lp \lim_{\rho,K\rightarrow \infty} \rho\sigma_{\hk m} \! > \! \ep \rp \= 1 \nonumber
\eeq
with probability one.

\subsubsection{Mean-square Convergence}
To show \eqref{lemma_mean_square_conv_equation} in mean-square sense, we need to first calculate quantities
$\lim\limits_{\rho,K\rightarrow\infty} E\lS \rho\sigma_{\hk m} \rS$ and
$\lim\limits_{\rho,K\rightarrow\infty} E\lS \rho^2\sigma_{\hk m}^2 \rS$.
The expectation of $ \rho\sigma_{\hk m}$  is simplified by
\beq
E\lS \rho\sigma_{\hk m} \rS
\d4&=&\d4 \int_{0}^{\infty} \lp \Pr\lp \rho\sigma > x \rp \rp^{K} dx   \nonumber \\
\d4&=&\d4 \frac{\rho}{(L-1)K}. \label{2.1.13}
\eeq
Then, $E\lS \lp \rho\sigma_{\hk m} \rp^2 \rS$ is formulated as
\beq
E\lS \lp \rho\sigma_{\hk m} \rp^2 \rS
\d4&=&\d4 E\lS \int_0^{\rho\sigma_{\hk m}} 2x dx \rS \nonumber \\
\d4&=&\d4 2\lp \frac{\rho}{(L-1)K} \rp^2 \label{2.1.16}
\eeq
where \eqref{2.1.16} is obtained by integration by parts.

Consequently, from \eqref{2.1.13} and \eqref{2.1.16}, as $\rho,K\rightarrow \infty$ while
maintaining $K\varpropto\rho^a$ with $a>1$,
the variance of $\rho\sigma_{\hk m}$, i.e.,
$\lim\limits_{\rho,K\rightarrow\infty} \! \lp \! E\!\lS \! \rho^2\sigma_{\hk m}^2 \!\rS \- E\!\lS \rho\sigma_{\hk m} \rS^2  \rp$ converges
\beq
\lim_{\rho,K\rightarrow\infty}\lp \frac{\rho^{1-a}}{L\-1} \rp^2 = 0. \nonumber
\eeq
This establishes
\beq
\lim_{\rho,K\rightarrow\infty} E\lS \la \rho\sigma_{\hk m} - E\lS \rho\sigma_{\hk m} \rS \ra^2 \rS = 0
\eeq
implying $\rho\sigma_{\hk m}\stackrel{m.s.}{\longrightarrow} 0$.
\end{proof}

Lemma \ref{lemma_mean_square_conv} readily characterize the convergence of the total degrees of freedom as follows.

\begin{theorem} \label{theorem_mean_suare_conv}
If the number of users $K$ in a cell increases faster than linearly with $\rho$, i.e., $\rho, K\rightarrow \infty$
in such a way that $K\varpropto\rho^a$ for $a>1$ and $a\in\R$,
the instantaneous degrees of freedom in \eqref{2.7} converges as
\beq
\lim_{\rho,K\rightarrow\infty}\frac{\sum\limits_{m=1}^{L}R_{\hk m}}{\log(\rho)} \stackrel{\substack{m.s.\\a.s.}}{=} L \label{theorem_mean_suare_conv_equation}
\eeq
where $M=1$ and $N= L\-1$.
\end{theorem}
\begin{proof}
The inner bound of the instantaneous degrees of freedom of the selected user $\hk m$ (by maximizing SINR)
yields
\beq
\lim_{\rho,K\rightarrow\infty}\frac{R_{\hk m}}{\log(\rho)}
\d4&\geq&\d4 \lim_{\rho,K\rightarrow\infty}\frac{\tilde{R}_{\hk m}}{\log(\rho)} \nonumber \\
\d4&\stackrel{\substack{m.s.\\a.s.}}{=}&\d4 \lim_{\rho\rightarrow\infty}
          \frac{ \log\lp 1\+\rho \la  \frac{\tbp_{\hk m}^*}{\| \tbp_{\hk m} \|_2}\bh_{\hk m, m} \ra^2 \rp }{\log(\rho)} \nonumber \\
\d4&\stackrel{a.s.}{\=}&\d4 1 \label{6.1.7.1}
\eeq
where we use the facts that $\rho\sigma_{\hk m}\stackrel{\substack{m.s. \\ a.s.}}{\longrightarrow} 0$
(i.e., Lemma \ref{lemma_mean_square_conv}) for $\tilde{R}_{\hk m}$ in \eqref{2.1.9.1}
and the quantity $| (\tbp_{\hk m}/\| \tbp_{\hk m} \|_2)^*  \bh_{\hk m, m} |^2$ is independent
of $\rho$ and $K$.
Notice that $\tbp_{\hk m}$ and  $\bh_{\hk m, m}$ are mutually independent and $\tbp_{\hk m}/\| \tbp_{\hk m} \|_2$
is isotropically distributed on the unit sphere.
Thus, $\la (\tbp_{\hk m}/\| \tbp_{\hk m} \|_2)^*\bh_{\hu m, m} \ra^2$ is exponentially distributed and
ensures
$\Pr\lp | (\tbp_{\hk m}/\| \tbp_{\hk m} \|_2) ^*\bh_{\hk m, m} |^2=0 \rp=0$ with probability one.
This fact leads to \eqref{6.1.7.1}.

Summing up the result in \eqref{6.1.7.1} from $m=1$ to $L$ yields the
achievable instantaneous degrees of freedom of $L$.
Recalling that $L$ is the maximum possible number of parallel streams in $L$-cell SIMO interference channel
concludes the proof.
\end{proof}

The result in \eqref{theorem_mean_suare_conv_equation} is strong in the sense that
the mode of convergence falls in the intersection of the two modes (i.e., almost sure (a.s.) convergence and
mean-square (m.s.) convergence).

\subsubsection*{Multiuser Diversity vs. Interference Alignment}
For the $L$-cell SIMO interference channel with $M=1$ and $N\geq 1$, the optimal degrees of freedom
achieved by the interference alignment (without user scheduling) can be formulated as \cite{Gou}
\beq
\bSig_d = \lim_{\rho, n \rightarrow \infty} \sum_{m=1}^L \frac{R_{m,n}(\rho)}{\log(\rho)}
\stackrel{a.s.}{=} \min(L,N) \label{6.8}
\eeq
where $n$ denotes the symbol extension index and $R_{m,n}(\rho)$ denotes the instantaneous rate at
the channel use $n$.
Notice that this characterizes the maximum instantaneous degrees of freedom obtained
by the interference alignment in \cite{Gou} without multiuser diversity.

When $N=L-1$, the optimal instantaneous degrees of freedom in \eqref{6.8} yields
\beq
\bSig_d \stackrel{a.s.}{=}L-1, \nonumber
\eeq
while the multiuser diversity system attains
\beq
\bSig_d \stackrel{\substack{m.s.\\a.s.}}{=}L \nonumber
\eeq
instantaneous degrees of freedom in both of a.s. and m.s. sense.
This strong mode of convergence is benefited by the user scheduling gain.
Notice that the interference alignment is based on the global notion of CSI at all nodes, while the
multiuser diversity system relies only on local CSI with one real number feedback from the receiver
to the transmitter.
The former utilizes infinite symbol extension in time or frequency domain with time-varying channel assumption, while
the later deals with infinite number users in the network with the constant channel coefficients.

Consequently, from Theorem \ref{theorem_mean_suare_conv} and \eqref{6.8}, when $N=L-1$ we make following crucial statement.
\begin{remark}
\emph{Utilizing multiuser diversity with local CSI provides at least additional $\frac{1}{L}$ instantaneous
degrees of freedom to each of the users in the $L$-cell downlink interference channel with $M=1$ and $N=L-1$}.
\end{remark}

\subsection{Average Degrees of Freedom Analysis} \label{section6.2}
The \emph{average} degrees of freedom without the notion of the convergence in random sequences
can now be formulated without difficulty. By taking expectation over all possible channel realizations,
the achievable average rate at user $\hk m$ with the maximum SINR user scheduling is denoted by
\beq
{\bar{R}}_{\hk m} = E\lS R_{\hk m} \rS \label{6.9}
\eeq
where $R_{\hk m}$ is given in \eqref{2.1.6}.
As can be seen from the theorem below, the user scaling law can be relaxed when the average throughput is
considered.

\begin{theorem} \label{theorem_avrage_dof}
If $K$ is linearly proportional to $\rho$ or faster than linear with $\rho$, i.e., $\rho, K\rightarrow \infty$ while
maintaining $K \varpropto \rho^a$ for $a \geq 1$ ($a\in\R$),
the average degrees of freedom of the maximum SINR user scheduler with $M=1$ and $N= L\-1$
is
\beq
\lim_{\rho,K\rightarrow\infty} \frac{ \sum\limits_{m=1}^L  \bar{R}_{\hk m}  }{\log(\rho)} = L. \label{theorem_avrage_dof_equation}
\eeq
\end{theorem}
\begin{proof}
The quantity in \eqref{6.9} is lower bounded by
\beq
\!\!\!\!\!\!\!\!\bar{R}_{\hk m} \!\d4&\geq&\!\d4 E\lS \tR_{\hk m} \rS \nonumber \\
\!\d4&\geq &\!\d4 E\!\lS \log\!\lp\!  \frac{\lA\tbp_{\hk m}\rA_2^2 \+ \rho | \tbp_{\hk m}^*\bh_{\hk m, m} |^2}
                                           {E\!\lS \lA\tbp_{\hk m}\rA_2^2 \rS \+ E\lS \rho \sigma_{\hk m}  \rS}\rp \rS \label{6.10}
\eeq
where in the second step we use $\rho \sigma_{\hk m} \geq 0$ and Jansen's inequality.

Plugging the result in \eqref{2.1.13} in \eqref{6.10} yields
\beq
\bar{R}_{\hk m} \!\geq\!
E\!\lS \log\!\lp\!  \frac{\lA\tbp_{\hk m}\rA_2^2 \+ \rho | \tbp_{\hk m}^*\tbh_{\hk m, m} |^2}
                                           {E\!\lS \lA\tbp_{\hk m}\rA_2^2 \rS \+ \frac{\rho}{(L\-1)K } }\rp \rS \label{6.11}
\eeq
Then, as $\rho, K$ tends to infinity, the average degrees of freedom of the r.h.s. of \eqref{6.11} converges to
\beq
1 - \lim_{\rho,K\rightarrow\infty}\frac{\log\lp E\lS \lA\tbp_{\hk m}\rA_2^2 \rS + \frac{\rho^{1-a}}{(L-1) } \rp}{\log(\rho)}=1 \nonumber
\eeq
as long as $a\geq 1$.

On the other hand, the outer bound of $\bar{R}_{\hk m}$ is obtained by ignoring interference term in \eqref{2.1.6}, i.e.,
\beq
\lim_{\rho\rightarrow \infty} E\lS \log\lp 1 + \rho \la \frac{\bp_{\hk m}^*}{\lA\bp_{\hk m}\rA_2}\bh_{\hk m, m} \ra^2 \rp / \log(\rho) \rS = 1. \nonumber
\eeq

Thus, $\lim\limits_{\rho,K\rightarrow\infty}\!\frac{\bar{R}_{\hk m}}{\log(\rho)}\=1$
and subsequently, $\lim\limits_{\rho,K\rightarrow\infty}\!\frac{\sum\limits_{m=1}^L \!\bar{R}_{\hk m}}{\log(\rho)}\=L$.
\end{proof}

Theorem \ref{theorem_avrage_dof} states that in order to achieve the average degrees of freedom
of $L$ for the $L$ selected users,
it is sufficient to increase $K$ like $K\propto\rho$ as $\rho\rightarrow \infty$.
We observe the user scaling law is relaxed compared to the case in Theorem \ref{theorem_mean_suare_conv}
so that it allows the linear increase.
However, the convergence in \eqref{theorem_avrage_dof_equation} does not include modes of the
convergence in random sequences, thereby,
the argument is quiet much weaker than \eqref{theorem_mean_suare_conv_equation}.
Theorem \ref{theorem_mean_suare_conv} implies Theorem \ref{theorem_avrage_dof}, while
Theorem \ref{theorem_avrage_dof} does not guarantee Theorem \ref{theorem_mean_suare_conv}.

\section{Conclusions} \label{section_conlusions}

We characterized the degrees of freedom for the multicell MIMO MAC consisting of
$L$ cells and $K$ users per cell with constant channel coefficients.
We presented a degrees of freedom outer bound and linear achievable schemes for a few cases
that obtain the optimal degrees of freedom.
The degrees of freedom outer bound showed that for virtual MIMO systems selecting transmitters with
partial message sharing (through perfect link) sometimes provided more degrees of freedom than employing
multiple distributed MIMO transmitters.
The characterized outer bound also provides insight into the degrees of freedom limit for the two-tier
heterogeneous network  where the network is composed of $(L-1)$ lower-tier cells each with single user
and one macrocell with $K$ users.
By simply characterizing the linear inner bound schemes, it was shown that the transmit zero
forcing and null space interference alignment achieve the optimal degrees of freedom for the two-cell case for arbitrary number of users.
We also verified that receive zero forcing achieves the optimal degrees of freedom
for $L>1$ and $K\geq 1$ without transmit and receive coordination.
The generalized null space interference alignment scheme was developed for various spatial dimension conditions
to provide $\beta$ interference free dimensions to each of users.
We also verified that the developed linear schemes indeed achieve the optimal degrees of freedom using the
minimum possible $M+N$ when assuming a single stream per user.
Exploiting multiuser diversity, we showed that the instantaneous degrees of freedom converges to $L$ in both
almost sure (a.s.) and mean-square (m.s.) sense for $L$-cell SIMO downlink interference channel with
$M=1$ and $N=L-1$. This exhibited clear comparison on the instantaneous degrees of freedom between
the multiuser diversity system and conventional interference alignment.

\appendices
\section{Lemma \ref{lemma_linear_independent_columns}} \label{appendix_linear_independent_columns}
\begin{lemma} \label{lemma_linear_independent_columns}
Given $\bA\in\C^{m\times n}$ and $\bB\in\C^{n \times l}$ with $n\geq\max(m,l)$ where $\bA$ and $\bB$ with i.i.d. are full rank
and are mutually independent,
$\bA\bB$ has $\rank(\bA\bB)\=\min(m,l)$ with probability one.
\end{lemma}
\begin{proof}
First, we assume $\min(m,l)\=m$ and decompose $\bB\= \lS \hbB \ \bB^{'} \rS$ where
$\hbB\in\C^{n \times m}$ is formed by taking the first $m$ columns of $\bB$ and
$\bB^{'}\in\C^{n \times (l\-m)}$ is composed of columns from $m\+1$ to $l$ columns
of $\bB$. Then, regarding $\rank(\bA\bB)$ we have
\beq
\rank\lp \bA\hbB \rp \leq \rank\lp \bA\bB\=[ \bA\hbB \ \bA\bB^{'}] \rp \leq \min(m,l)\=m \label{ap1.0}.
\eeq
Note that when $\min(m,l)\=l$, we only need to consider the matrix $\bB^*\bA^*$,
and it is handled similarly to the case $\min(m,l)\=m$. Thus, we omit the case $\min(m,l)\=l$ and
focus on $\min(m,l)\=m$.

We further decompose $\bA \= \lS \bbA \ \tbA \rS$ and
$\hbB^* = \lS \bbB \ \tbB \rS$ where $\bbA\in\C^{m\times m}$ and
$\bbB\in\C^{m\times m}$ are formed by taking the first $m$ columns
of $\bA$ and $\hbB^*$, respectively, and $\tbA\in\C^{m\times (n\-m)}$
and $\tbB\in\C^{m\times (n\-m)}$ are submatrices corresponding to columns
from $m\+1$ to $n$ of $\bA$ and $\hbB^*$, respectively.

We claim $\Pr\big( \big| \det\big(\bA\hbB\big) \big| >0  \big)\=1$.
The claim is verified by providing the converse, i.e., $\Pr\big( \det\big( \bA\hbB \big)  \= 0 \big)\=0$.
Since $\bA$ and $\hbB$ are drawn from i.i.d. continuous distributions, their principal submatrices $\bbA$ and $\bbB^*$
(square matrices) are full rank matrices ($\rank(\bbA)\=m$ and $\rank(\bbB^*)\=m$) almost surely.
Now, we have
\beq
\Pr \lp  \det\lp \bA\hbB \rp  \= 0 \rp
\d4&=&\d4 \Pr \lp  \det\lp \bbA\bbB^* + \tbA\tbB^* \rp  \= 0\rp  \nonumber \\
\d4&=&\d4 \Pr \lp  \det\lp \bbA\bbB^* \rp    \det\lp \bI_{m}
     \+ \lp \bbA\bbB^* \rp^{-1} \tbA\tbB^*  \rp  \=0 \rp  \nonumber \\
\d4&=&\d4 \Pr \lp  \ls \det\lp \bbA\bbB^* \rp\=0 \rs \cup \ls \det\lp \bI_{m}
     \+ \lp \bbA\bbB^* \rp^{-1} \tbA\tbB^* \rp\=0 \rs \rp  \label{ap1.1}.
\eeq
By using the fact that both $\bbA\bbB^*$ and
$\bI_{m} \+ \lp \bbA\bbB^* \rp^{-1} \tbA\tbB^*$ are invertible $m \times m$ matrices,
from \eqref{ap1.1} we obtain
\beq
\Pr \lp  \det\lp \bA\hbB \rp  \= 0 \rp \leq \Pr \lp  \det\lp \bbA\bbB^* \rp\=0 \rp
     \+ \Pr \lp  \det\lp \bI_{m} \+ \lp \bbA\bbB^* \rp^{-1} \tbA\tbB^* \rp\=0 \rp\=0 \nonumber
\eeq
Consequently, we get $\Pr\big( \det\big( \bA\hbB \big)  \= 0 \big)\=0$ implying that
the left hand side (l.h.s.) of \eqref{ap1.0} is $\rank\big(\bA\hbB\big)\=m$. This concludes the proof.
\end{proof}

\section{Proof of Theorem \ref{theorem_RX_ZF_optimal_dof}} \label{appendix_theorem_RX_ZF_optimal_dof}
The converse is checked by plugging $M=\beta$ and $N\=LK\beta$ in \eqref{corollary_homogenous_outerbound_equation},
which in turn yields
\beq
\bSig_d \d4&\leq&\d4 \min\lp KL\beta, KL^2\beta, \frac{(KL)^2(L-1)\beta}{K+L-1}, \frac{(KL)^2\beta}{K+L-1} \rp \nonumber \\
\d4&\leq&\d4 \min\lp KL\beta,  \frac{KL}{K+L-1} KL\beta\rp = KL\beta. \nonumber
\eeq
The last equality follows from the fact that $KL\geq K+L-1$ for $K,L\geq 1$.

Inner bound is argued by using receive zero forcing.
When $N=LK\beta$ and $M=\beta$, base station $m$ chooses a null space plane $\bP_m\in\C^{K\beta\times LK\beta}$ such that
\beq
\span\lp \bP_m^T \rp \subset null\lp \lS \bH^{[m,1\cK]} \cdots \bH^{[m,(m-1)\cK]} \ \bH^{[m,(m+1)\cK]} \cdots \bH^{[m,L\cK]}\rS^T \rp \label{4.1}
\eeq
where $\bH^{[m,l\cK]}=\lS \bH_{m,l1} \cdots \bH_{m,lK} \rS\in\C^{LK\beta\times K\beta}$.
Since $\big[ \bH^{[m,1\cK]} \cdots \bH^{[m,(m-1)\cK]} \ \bH^{[m,(m+1)\cK]} \cdots \bH^{[m,L\cK]} \big]^T\in\C^{(L\-1)K\beta\times LK\beta}$,
$\bP_m$ that satisfies \eqref{4.1} with $\rank(\bP_m^T)=K\beta$ can be found with probability one.
Postprocessing $\by_m$ in \eqref{4.0} with $\bP_m$ returns
\beq
\tby_m\=\sum_{k=1}^K \bP_m\bH_{m,mk}\bT_{mk}\bs_{mk}\+\bP_m\bz_m = \bP_m\bG_m\tbs_m + \tbz_m.\nonumber
\eeq
where $\bG_m=[\bH_{m,m1}\bT_{m1} \cdots \bH_{m,mK}\bT_{mK}]\in\C^{LK\beta\times K\beta}$,
$\tbz_m=\bP_m\bz_m$, and $\tbs_m=\lS \bs_{m1}^T \cdots \bs_{mK}^T \rS^T$.
Here, $\bT_{mk}\in\C^{\beta\times\beta}$ can be arbitrary with $\rank\lp\bT_{m}\rp\=\beta$.
Without loss of generality, $\bT_{mk}$ can be taken to be $\bT_{mk}=\bI_{\beta}$.
As observed in the proof of Theorem \ref{theorem_Two_cell_MAC_optimal_dof},
$\bP_m$ and $\bG_m$ are mutually independent and $\bP_m\bG_m$ spans a $K\beta$-dimensional space with probability one.
This ensures the achievability of $LK\beta$ degrees of freedom for $L$-cell and $K$-user MIMO MAC.

\section{Proof of Theorem \ref{theorem_geometric_multiplicity}} \label{appendix_theorem_geometric_multiplicity}
Assume $\ls \bA_{1 }, \ldots, \bA_{K} \rs$ has $\gam$ null space multiplicity.
Since the matrices $\ls \lS \bA_{k} \ \bB_k \rS \rs_{k\in\cK}$ are nondegenerate, the $\gam$ and $\mu$ do not depend on the
choice of $\gam$-tuple matrix set.
Thus,  without loss of generality, we assume a $\gam$-tuple combination
$\{\bA_{i }\}_{i=1}^{\gam}$.
Set $\bGam_1=\bB_{1 }$.
Then, it is clear that $\bA_{1 }^*\bGam_1\=\bzero$.
Let $\bZ_{2}\in\C^{(n-m)\times(n-2m)}$ be an orthonormal basis of $null(\bA_{2}^*\bGam_1)$ and
denote $\bGam_2=\bGam_1\bZ_2$.
Since $\bA_{1}^*\bGam_2=\bzero$ and $\bA_{2}^*\bGam_2=\bzero$, $\bGam_2$ is in
$null(\bA_{1}^*)\cap null(\bA_{2}^*)$.
In the same manner, $\bGam_{i}$ for $i>2$ is designed with the
recursion
\beq
\bGam_{i}=\bGam_{i\-1}\bZ_{i} \label{ap3.2}
\eeq
where $\bZ_{i}$ is an orthonormal basis of $null(\bA_{i}^*\bGam_{i\-1})$.
Then, after $\gam$ times of recursions, we have
$\bGam_{\gam}=\bGam_{\gam\-1}\bZ_{\gam}\in \C^{n\times(n-\gam m)}$,
and since $\bA_{\gam\-1}^*\bGam_{\gam}\=\bzero$ and $\bA_{\gam }^*\bGam_{\gam}\=\bzero$,
we have
\beq
\bGam_{\gam}\subset\bigcap\limits_{i=1}^{\gam} null(\bA_{i })\label{ap3.1}.
\eeq
The existence of $\bGam_{\gam}$ in \eqref{ap3.1} (i.e., the existence of $\bZ_{\gam}$) is
therefore ensured if $n-\gam m \geq 1$, i.e.,
$\gam\leq \frac{n-1}{m}$
which is equivalent to
\beq
\gam=\left\lfloor \frac{n-1}{m} \right\rfloor=\left\lceil\frac{n-m}{m}\right\rceil.
\eeq
Notice that the result does not depend on the choice of $\gam$-tuple matrix set.
Since $\gam$ can not exceed $K$, $\gam$ is characterized as $\gam=\min\lp \lceil\frac{n-m}{m}\rceil, K \rp$.
Note that $\gam$ is the maximum possible integer such that $n-\gam m \geq 1$ implying
$\mu=\rank(\bGam_{\gam})$ is given by
\beq
\mu=n-\gam m
\eeq
and $1\leq \mu \leq m$.
This concludes the proof.

\section{Extension to $L>2$ Case} \label{appendix_NIA_L_geq_2}
When $L>2$, there are total $(L-1)K\beta$ out-of-cell interference streams.
We need to align $(L-1)K\beta$ interference streams to the lower dimensional subspace than $K\beta$-dimensional subspace
to provide $\beta$ interference free dimensions for each of users.
Since the dimension of the out-of-cell interference streams is larger than the dimension available at the reciever
(i.e., $K\beta < (L-1)K\beta$), direct extension of the framework for $L=2$ case seems not to work.
To solve this problem, we consider to aggregate out-of-cell interference channels.

Given $\ls \bH_{m,\ell k} \rs_{\ell\in\cL\bsh m, k\in\cK}$, channel aggregation is performed by collecting
$(L-1)$ out-of-cell interference channels such that
\beq
\tbH_{m,\barm k} = \lS \bH_{m,1k} \ \cdots \bH_{m,(m-1)k} \ \bH_{m,(m+1)k} \cdots \bH_{m,Lk} \rS \nonumber
\eeq
where $\tbH_{m,\barm k}\in\C^{N\times (L-1)M}$.
This aggregation results in total $K$ aggregated out-of-cell interference channels $\ls \tbH_{m,\barm k} \rs_{k\in\cK}$.
Then, the geometric multiplicity $\gam$ of $\ls \tbH_{m,\barm k} \rs_{k\in\cK}$ is expressed as
\beq
\gam=\min\Big(  \left\lceil \frac{N-(L-1)M}{(L-1)M} \right\rceil,K  \Big). \label{5.10.1}
\eeq
In \eqref{5.10.1}, we make the assumption that $N>(L-1)M$ (i.e., $1\leq \gam\leq K$).

Now consider full rank matrices $\ls \lS \tbH_{m,\barm k} \ \tbN_{m,\barm k} \rS \rs_{k\in\cK}$ where
$\tbN_{m,\barm k}\in \C^{N\times (N-(L-1)M)}$.
Under the same definition for the index set $\Pi_k=\ls \pi_i \rs_{i=k}^{\gam+k-1}$ as in \eqref{5.4.1},
the intersection null space is denoted by $\tbN_{m,\barm}^{(k)}\in\C^{N\times \mu}$, i.e.,
\beq
\span\Big( \tbN_{m,\barm}^{(k)} \Big) \subset \bigcap\limits_{i=k}^{\gam+k-1} ran\Big( \tbN_{m,\pi_i} \Big).
\eeq
Then, following the same framework for designing $\bP_m$ as $L=2$ case,
when $1\leq \gam\leq K-1$, $\bP_m$ is formed by
\beq
\bP_m \= \bPi \lS \tbN_{m,\barm}^{(1)} \  \tbN_{m,\barm}^{(2)} \ \cdots \tbN_{m,\barm}^{(K)} \rS^* \label{5.11}
\eeq
with $N=(L-1)\gam M + \beta$.
When $\gam=K$, we have $\tbN_{m,\barm}^{(1)}\in\C^{N\times K\beta}$ and
\beq
\bP_m=\bPi \tbN_{m,\barm}^{(1)^{\scriptstyle *}} \label{5.11.1}
\eeq
which is possible if $N=(L-1)\gam M + K\beta$.
Now, given $\bP_m$ in \eqref{5.11} and \eqref{5.11.1}, the projected out-of-cell interference channel $\bP_m \bH_{m,\barm k}$ satisfies
$\rank ( \bP_m\bH_{m,\barm k} )=(K-\gam)\beta$ for $k\in\cK$, $\barm\in\cL\bsh m$.
Now, under the zero out-of-cell interference constraint $\span (\bW_{\barm k})\subset null( \bP_m \bH_{m,\barm k} )$, we
must have
\beq
M=(L-1)(K-\gam)\beta + \beta. \label{5.12}
\eeq

\bibliographystyle{IEEEtran}
\bibliography{IEEEabrv,Journal_IA_null}

\newpage

\begin{figure}[!htb]
\centering
\includegraphics[width=9cm, height=9cm]{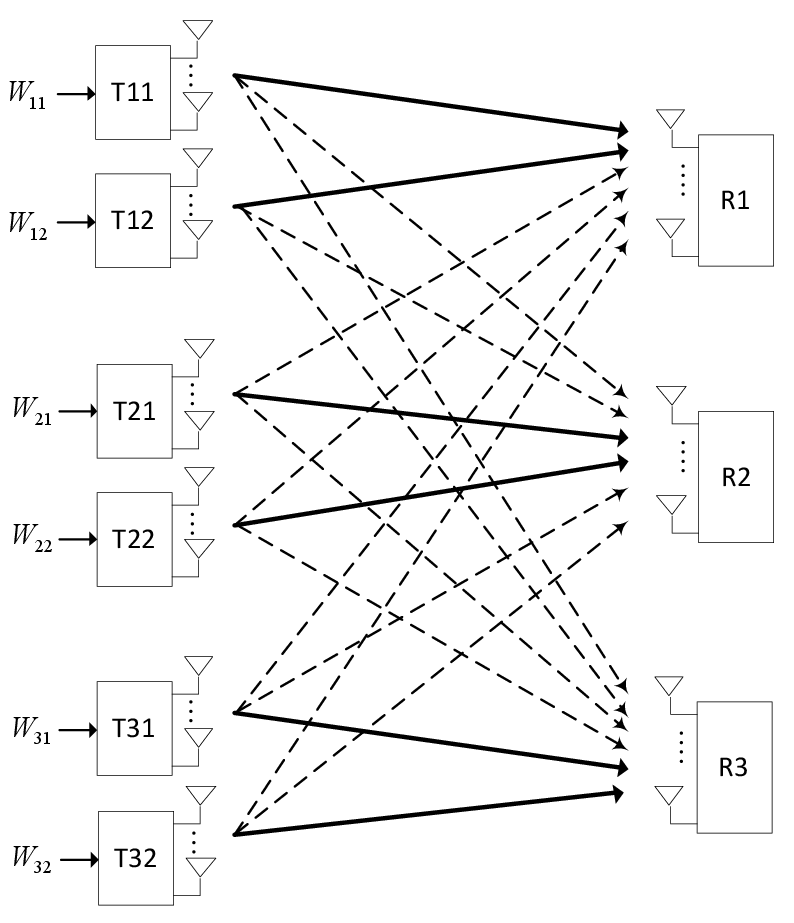}
\caption{Multicell MIMO MAC with $L=3$ and $K=2$.}
\label{Fig1}
\end{figure}

\begin{figure}[htb]
    \centering
    \includegraphics[width=10cm, height=9.5cm]{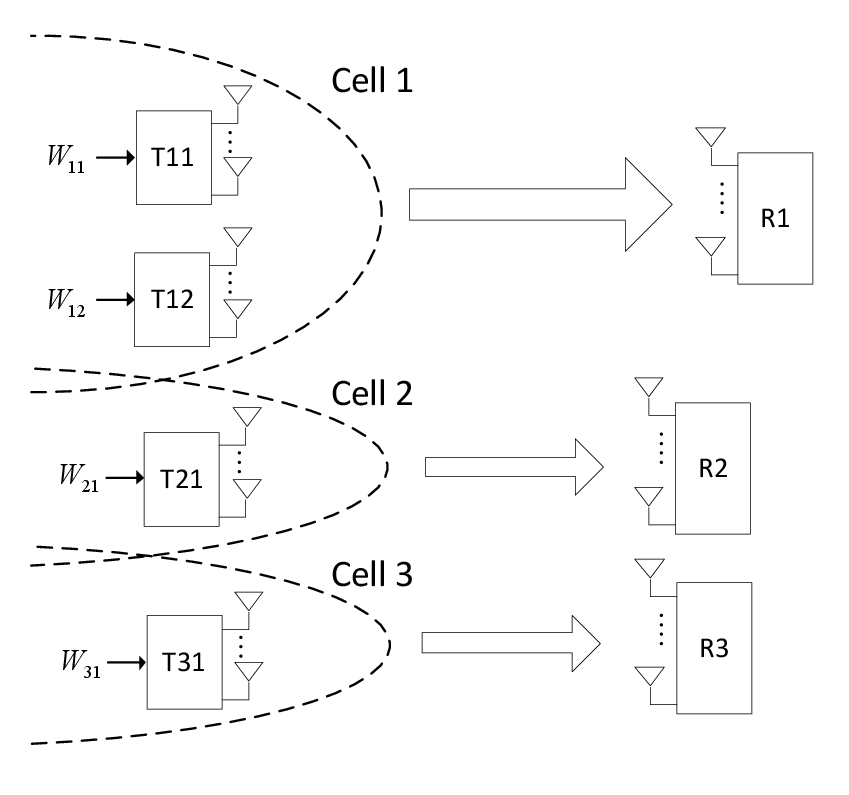}
    \caption{$(1,2)$ MAC-IC uplink HetNet consisting of a $2$-user MIMO MAC (i.e., cell $1$) and
             $2$-user MIMO interference channel (i.e., cell $2$ and $3$).}
    \label{Fig6}
\end{figure}

\begin{figure}[htb]
    \centering
    \includegraphics[width=10cm, height=10cm]{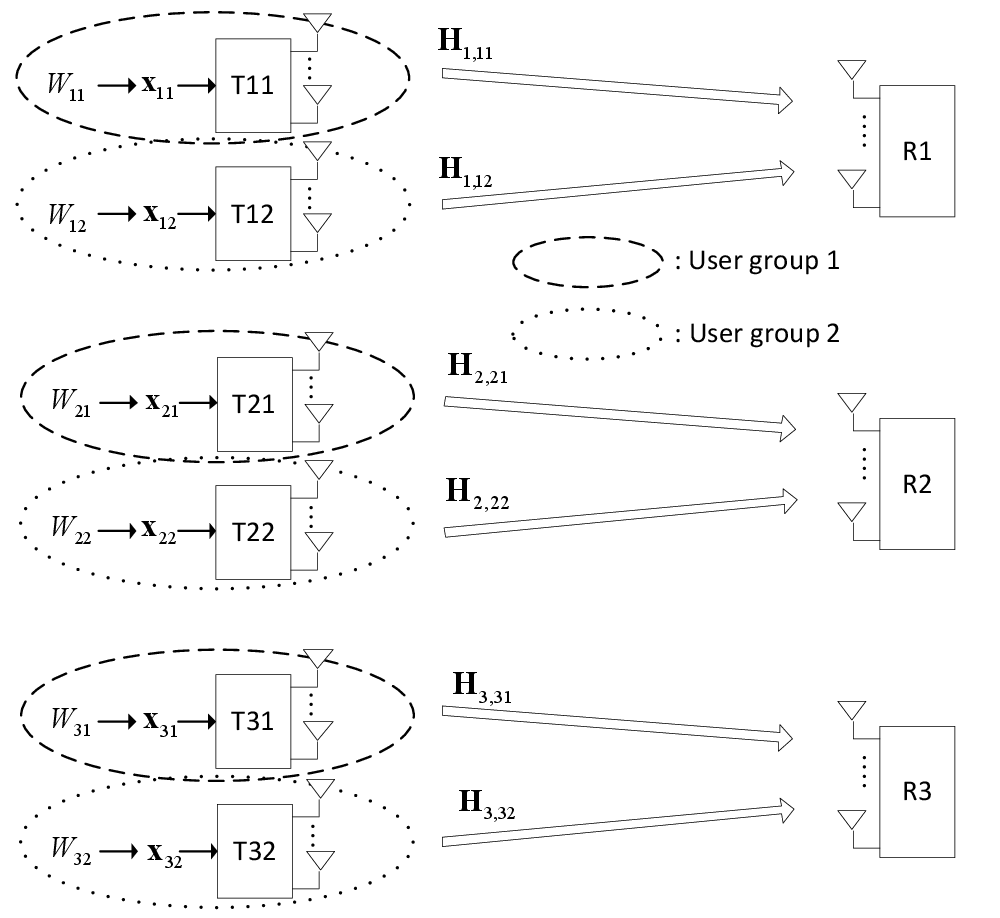}
    \caption{User grouping strategy for $L=3$ and $K=2$ MIMO MAC}
    \label{Fig3a}
\end{figure}

\begin{figure}[htb]
    \centering
    \includegraphics[width=10cm, height=10cm]{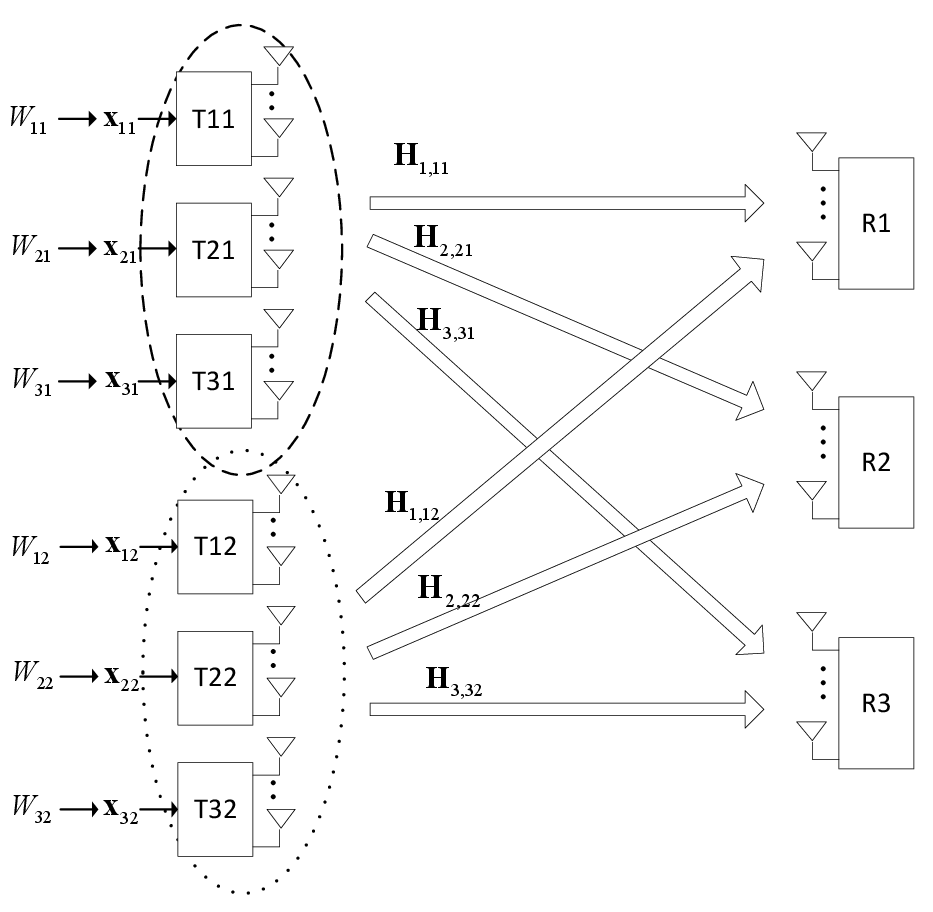}
    \caption{Conversion to distributed $2 \times 3$ homogeneous MIMO X channel}
    \label{Fig3b}
\end{figure}


\begin{figure}[htb]
    \subfigure[]{
    \epsfig{figure=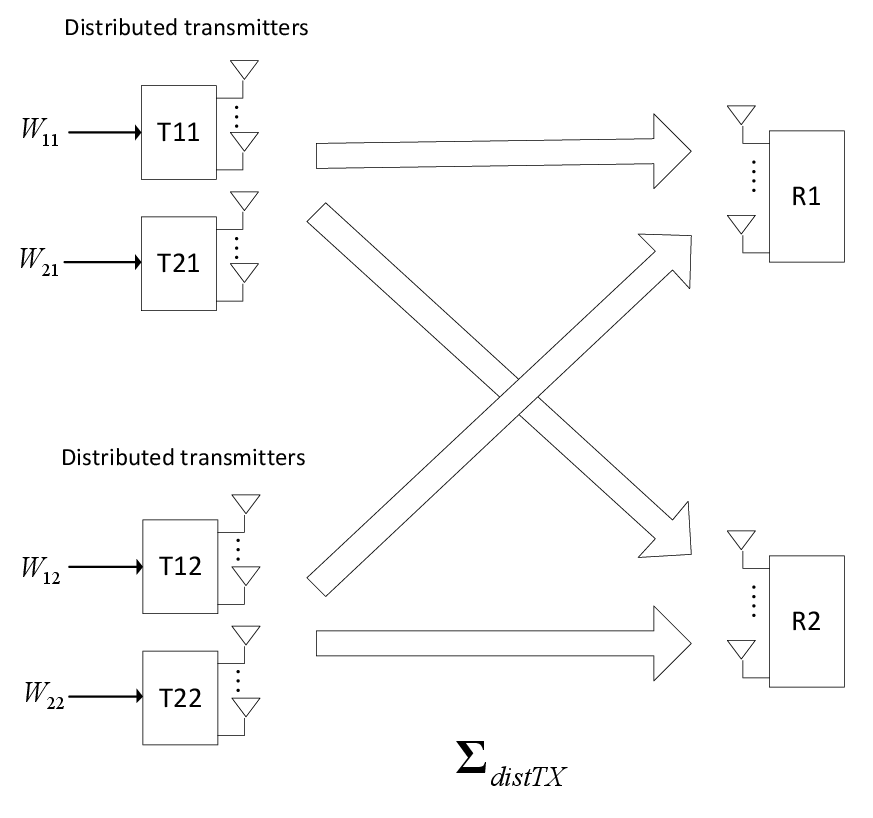, width=0.45\textwidth, height=8.5cm}}
    \subfigure[]{
    \epsfig{figure=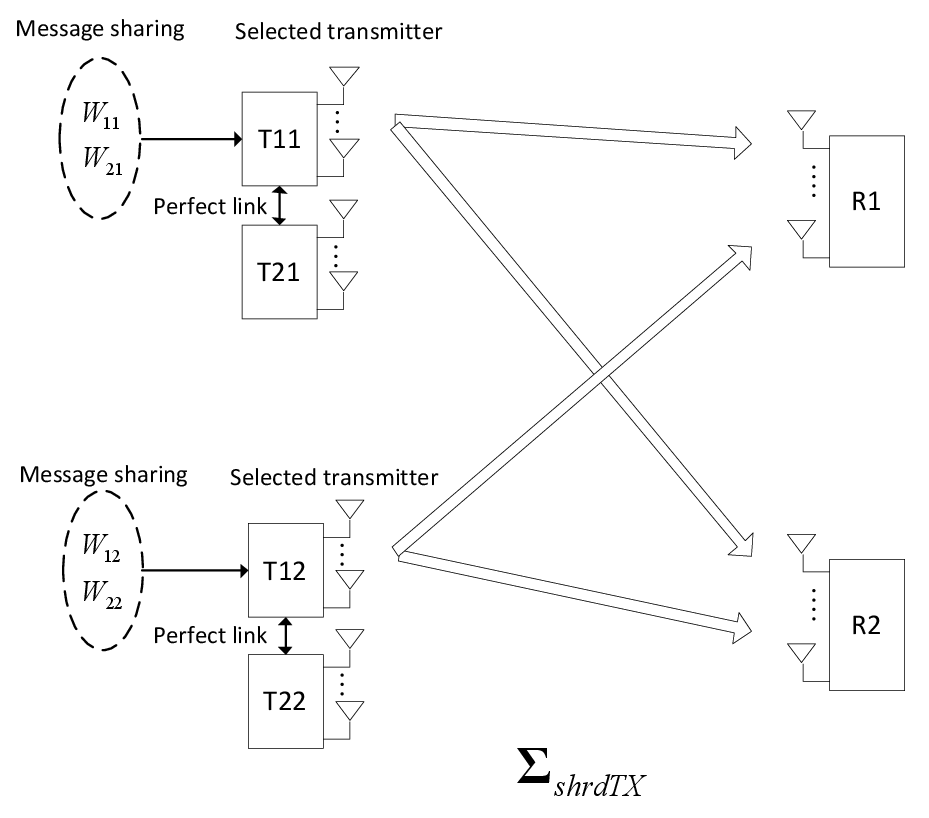, width=0.48\textwidth, height=8.5cm}}
    \caption{(a) Multiple distributed transmission ($L=2$ and $K=2$). (b) Selected and shared transmission through perfect links ($L=2$ and $K=2$)}
    \label{Fig5}
\end{figure}


\end{document}